%% file: FQCA-beta.tex
\documentclass[aps,twocolumn,superscriptaddress,10pt]{revtex4-2}

\usepackage{times}
\usepackage{comment}
\usepackage{scalerel}

\usepackage{epsfig,amsmath,amssymb,mathrsfs,mathtools,amscd}
\usepackage[linkcolor=red]{hyperref}

\usepackage{cleveref}
\usepackage{soul}
\usepackage{enumerate}
\usepackage{amsmath} 
\usepackage{amsthm}
\usepackage{cancel}
\usepackage{verbatim} 
\usepackage{mathtools} 
\usepackage{float} 
\usepackage{tikz}
\usepackage{tikzit}
\usetikzlibrary{decorations.text}
\input{tikzstyle.tex}
\input{sample.tikzstyles}

\newtheorem{lemma}{Lemma}
\newtheorem{definition}{Definition}

\newtheorem{proposition}{Proposition}

\newtheorem{theorem}{Theorem}
\newtheorem{corollary}{Corollary}
\theoremstyle{remark}
\newtheorem{remark}{Remark}
\def\A{\ensuremath{\mathsf{A}}}
\def\ind#1{\ensuremath{\operatorname{ind}[#1]}}
\def\dim#1{\ensuremath{\operatorname{dim}[#1]}}
\def\transf#1{\ensuremath{\mathcal{#1}}}
\def\Cl#1#2{\ensuremath{\mathit{Cl}_1(#1\lvert #2)}}
\def\tG{\transf G}
\def\tT{\transf T}

\def\tS{\transf S}

\def\tU{\transf U}

\def\tF{\transf F}
\def\tK{\transf K}
\def\tJ{\transf J}
\def\tM{\transf M}
\def\tI{\transf I}
\def\Fimpl{$\tF$-Implementability}
\def\Mimpl{$\tM$-Implementability}
\def\Fequiv{$\tF$-Equivalence}

\def\comm#1#2{\ensuremath{[#1,#2]}}
\def\acomm#1#2{\ensuremath{\{#1,#2\}}}
\def\m#1{\ensuremath{\mathsf{#1}}}
\def\gc#1#2{\ensuremath{\{[#1,#2]\}}}

\def\Aqloc{\ensuremath{\mathsf{A}(\mathbb{Z})}}

\def\gt{\ensuremath{\boxtimes}}

\def\Cl#1#2{\mathit{C}\ell_1(#1\lvert #2)}
\mathchardef\minus="002D

\def\<{\langle}
\def\>{\rangle}
 \def\ket#1{| #1 \rangle}
\def\bra#1{\langle #1 |}
\def\ketbra#1#2{| #1 \rangle \langle#2 |}

\def\Z{\mathbb Z}

\def\EL{\ensuremath{\mathcal{E}_L}}
\def\ER{\ensuremath{\mathcal{E}_R}}
\def\EC{\ensuremath{\mathcal{E}_C}}

\def\M#1#2{\operatorname{Mat}(\mathbb{C}^{#1\lvert #2})}
\def\Z2{\mathbb{Z}_2}
\DeclareMathOperator*{\bigboxtimes}{\scalerel*{\boxtimes}{\sum}}

\definecolor{AO}{rgb}{0.0, 0.5, 0.0}

\newcommand\Item[1][]{%
  \ifx\relax#1\relax  \item \else \item[#1] \fi
  \abovedisplayskip=0pt\abovedisplayshortskip=0pt~\vspace*{-\baselineskip}}

\newcommand{\Eq}{Eq.~}
\newcommand{\Eqs}{Eqs.~}

\begin{document}

\title{Fermionic cellular automata in one dimension}

\author{Lorenzo S.
  \surname{Trezzini}} \email[]{lorenzosiro.trezzini01@universitadipavia.it}
\affiliation{Dipartimento di Fisica dell'Universit\`a di Pavia, via
  Bassi 6, 27100 Pavia} \affiliation{Istituto Nazionale di Fisica
  Nucleare, Gruppo IV, via Bassi 6, 27100 Pavia} 
  \author{Matteo
  \surname{Lugli}} 
  \author{Paolo 
  \surname{Meda}} \email[]{paolo.meda@unipv.it}
\affiliation{Dipartimento di Fisica dell'Universit\`a di Pavia, via
  Bassi 6, 27100 Pavia} \affiliation{Istituto Nazionale di Fisica
  Nucleare, Gruppo IV, via Bassi 6, 27100 Pavia} 
\author{Alessandro 
  \surname{Bisio}} \email[]{alessandro.bisio@unipv.it}
\affiliation{Dipartimento di Fisica dell'Universit\`a di Pavia, via
  Bassi 6, 27100 Pavia} \affiliation{Istituto Nazionale di Fisica
  Nucleare, Gruppo IV, via Bassi 6, 27100 Pavia} 
\author{Paolo
  \surname{Perinotti}} \email[]{paolo.perinotti@unipv.it}
\affiliation{Dipartimento di Fisica dell'Universit\`a di Pavia, via
  Bassi 6, 27100 Pavia} \affiliation{Istituto Nazionale di Fisica
  Nucleare, Gruppo IV, via Bassi 6, 27100 Pavia} 
  \author{Alessandro 
  \surname{Tosini}} \email[]{alessandro.tosini@unipv.it}
\affiliation{Dipartimento di Fisica dell'Universit\`a di Pavia, via
  Bassi 6, 27100 Pavia} \affiliation{Istituto Nazionale di Fisica
  Nucleare, Gruppo IV, via Bassi 6, 27100 Pavia}

\begin{abstract} 
  We consider quantum cellular automata for one-dimensional chains of
  Fermionic modes and study their implementability as finite depth
  quantum circuits. Fermionic automata have been classified in terms
  of an index modulo circuits and the addition of ancillary
  systems. We strengthen this result removing the ancilla degrees of
  freedom in defining the equivalence classes. A complete
  characterization of nearest-neighbours automata is given. A class of
  Fermionic automata is found which cannot be expressed in terms of
  single mode and controlled-phase gates composed with shifts, as is
  the case for qubit cellular automata.
\end{abstract} 

\maketitle

\section{Introduction}
Besides the perspectives in enhancing classical computation and
communication protocols, one of the major applications of quantum computers 
is to provide a tool for the
simulation of quantum physical systems. A natural architecture toward
this goal is the quantum generalization of cellular automata, whose
power in simulating classical dynamical systems has been largely
proved~\cite{Chopard_Droz_1998,VICHNIAC198496}. Cellular automata were
envisioned~\cite{neumann1966theory} as a mesh of locally interconnected 
finite state machines, called cells. The cells change their states synchronously
depending on their state and the nearby cells' one at each time step,
thus embodying the \emph{locality of interactions} typical of physical systems. 
Other relevant properties with a physical counterpart, 
such as \emph{reversibility} can be programmed by
suitably choosing the local update rule.

If cellular automata bridge complex dynamical systems and computation,
quantum cellular automata (QCAs) ~\cite{schumacher2004reversible} are
the natural setting for the efficient
simulation of quantum systems, and at the same time they provide a model of 
universal quantum
computation~\cite{PhysRevLett.97.020502,PhysRevA.72.022301}. Tensor
networks~\cite{doi:10.1080/14789940801912366,Evenbly-Vidal-tensor-network},
the state of the art of nowadays quantum simulations, have also been
linked~\cite{MPU-QCA-2017,PhysRevLett.125.190402} to QCAs in the
notion of matrix-product unitaries, tensor networks describing time
evolution and unitary symmetries of quantum systems. Following the
growing interest on the computational features of non-abelian matter
fields~\cite{KITAEV20032,Kitaev_2001,Nature-non-abelian-qunatum-info,PhysRevLett.103.020506},
cellular automata, initially introduced for qudits, have been
generalized to Fermionic
systems~\cite{PhysRevB.99.085115,Perinotti2020cellularautomatain,Farrelly2020reviewofquantum},
taking care to keep the parallelism with the literature on tensor
networks~\cite{Piroli_2021}.

As quantum simulators, QCAs have been successful in describing several
phenomena of many-body quantum systems with local interactions, such
as Floquet systems, and in classification of the topological phases of
matter~\cite{PhysRevX.6.041070,PhysRevB.103.064302,PRXQuantum.3.030326,Stephen2019subsystem},
which has become a concept of fundamental importance in condensed
matter. Efforts have also been given on identifying classes of QCAs
that admit of a simple particle interpretation, that is automata whose
dynamics can be interpreted in terms of the propagation and scattering
of bosonic and Fermionic particles in quantum field theories or many
body systems~\cite{PhysRevA.90.062106,PhysRevA.97.032132,Arrighi-qed,BISIO2016177}. A primary intent underlying discrete quantum simulations is on one hand the development of methods to access their continuum limit and effective dynamics
\cite{PhysRevX.11.011020,PhysRevLett.123.050503,PhysRevLett.97.157202,PhysRevA.92.022330,PhysRevLett.126.250503,PhysRevA.93.062334}, and on the other hand to model the dynamics of quantum systems in an emergent 
space-time~\cite{PhysRevA.90.062106,refId0,PhysRevA.94.042120,10.1063/1.5144731}.

Despite several formal definitions of QCAs over the
years~\cite{Watrous-qca,Farrelly2020reviewofquantum}, a systematic
exploration of the general properties of such systems and of their
potential for computational applications was started quite
recently. In Ref.~\cite{schumacher2004reversible} an
algebraic axiomatic approach is proposed  to describe a grid of finite dimensional
quantum systems with translation symmetry undergoing a reversible
discrete time evolution driven by a local rule. At each site local
operators (in place of local Hilbert spaces) are considered, whose
$C^*$-algebra structure can be extended to the whole infinite grid
of quantum systems. Among various alternative approaches, it is worth mentioning a
constructive method to obtain QCAs, consisting in partitioning the
system into blocks of cells, applying blockwise unitary
transformations, and possibly iterating such operations. A major
result is the proof~\cite{ARRIGHI2011372} that any QCA can be implemented in the above
constructive way with the aid of auxiliary systems. The main advantage of the partitioned scheme is that
the global unitary evolution can be obtained as the combination of
many local unitary gates involving few nearby systems, with a clear
operational (and experimental) interpretation. Most of the results 
on universality of cellular automata~\cite{PhysRevLett.97.020502,PhysRevA.72.022301,Arrighi-partitioned-univertsal} have been obtained within the partitioned scenario.

The partitioned scheme and the result of Ref.~\cite{ARRIGHI2011372}
seed one of the core questions on the structure of QCAs that is their
implementability by quantum circuits of finite depth. The answer to
this question has been given in one space dimension, where QCAs are
classified in terms of a topological invariant named \emph{index}, a
quantity that summarises the local invariants of any QCA, i.e. any quantity 
that can be computed knowing the local evolution of a finite number of cells,
and in intuitive terms represents the balance of information flow in the 
two possible directions on the one-dimensional lattice. The classification 
is in the following sense~\cite{Werner-index}: if two automata have the same 
index, then necessarily one can be obtained from the other by postponing
a finite depth circuit. Having the same index establishes an \emph{equivalence}
relation. Moreover, a QCA is implementable by a circuit if
and only if its index is equal to one. Recently, preliminary results
have been obtained in the generalization of index theory for lattices
having spatial dimension larger that
one~\cite{Freedman-Hastings-classification,Freedman-Haah-Hastings-group-structure,pizzamiglio2024classification},
and for Fermionic systems~\cite{PhysRevB.99.085115}---which is the focus of the 
present manuscript. The present knowledge about the Fermionic index highlights two 
main differences with respect to the qudits one.  i) In the case of a qudit QCA 
the index is a rational number, while it may take
irrational values in the Fermionic case. ii)
The classification of Fermionic Cellular Automata (FCAs) is based on a weaker
notion of equivalence, denoted \emph{stable equivalence}: two
Fermionic automata are stably equivalent if upon appropriately
enlarging the Hilbert space by appending ancillary Fermionic
degrees of freedom that do not evolve, one can find a finite depth circuit 
connecting them. In other words to connect cellular automata in the same equivalence 
class one needs to apply circuits that involve not only the physical systems at hand,
but also auxiliary systems which will be left invariant by the evolution.

In the first part of this work it is shown how, also in the Fermionic case, 
the classification of FCAs~\cite{PhysRevB.99.085115} in one dimension can 
actually be established in terms of equivalence classes modulo finite depth
Fermionic circuits, without the need of ancillary degrees of
freedom. The proof is based on the procedure of ancilla removal introduced
in Ref.~\cite{Freedman-Haah-Hastings-group-structure}. A complete index 
theory for Fermionic automata in one space dimension follows,
were automata with the same index are equivalent modulo circuits
involving only the physical systems, like in the quantum case.

The second result of this work is an explicit characterization of all
FCAs in one space dimension and with nearest-neighbouring
interaction. This is done by providing the explicit form of the most
general FCA having unit index and noticing that, as in the
quantum case, a generic automaton is obtained by composition of a unit
index instance with a Fermionic shift. The set of shifts for local
Fermionic modes is strictly larger than the qubits counterpart, due to
the \emph{parity superselection rule}. The Fermionic local algebra 
is indeed graded, and its even and odd modules can
be independently translated on the
lattice~\cite{PhysRevB.99.085115}. This leads to the phenomenon of
irrational index, corresponding to a net ``flow'' amounting to less than a cell per step, 
and yet non balanced, which is not achievable for ``quantum information''. Besides 
Fermionic shifts without quantum counterpart, we find solely Fermionic instances also
in the equivalence classes corresponding to a rational index. The full
characterization of nearest-neighbour qubits QCAs was derived in
Ref.~\cite{schumacher2004reversible}, with all QCAs corresponding to local
unitaries or controlled-phase gates, while the Fermionic local rules were 
classified in in Ref.~\cite{PhysRevA.98.052337} with a technique that gives no insight
in their structure. Here we show that Fermionic systems support a non-trivial FCA, 
that we deem \emph{forking} automaton, which lies outside the above classes generated 
by local gates together with d-phases.

\section{Algebra of a Fermionic chain}

This section is devoted to the algebraic description of a chain of
finitely many \emph{local Fermionic modes} arranged on the sites $x$ of a
one-dimensional lattice $\mathcal{L}\cong\mathbb{Z}$, as for Kitaev
chains of spinless Fermionic modes~\cite{Kitaev_2001}. The
computational properties of spinless Fermions were first studied in
Ref.~\cite{BRAVYI2002210} and later further explored in several
works~\cite{DAriano_2014,doi:10.1142/S0217751X14300257,Szalay_2021}.

Let $\mathcal{H}$ be a (one-particle) Hilbert space and consider the Fermi-Fock space
\begin{equation}
	\label{eq:fock-der}
	\begin{aligned}
		\mathfrak{F}_-(\mathcal{H}) &\coloneqq \overline{\bigoplus_{n \geq 0} \operatorname{ASym} \mathcal{H}^{\otimes n}} \\
		&= \mathbb{C} \oplus \mathcal{H} \oplus \left( \operatorname{ASym} (\mathcal{H} \otimes \mathcal{H}) \right) \oplus \ldots,
	\end{aligned}
\end{equation}
where $\operatorname{ASym} T$ denotes the antisymmetrization operator
acting on tensor products, while the overline denotes the completion
of the tensor product of copies of $\mathcal{H}$. Let
$\phi_x^\dagger$, $\phi_x$ be respectively the creation and
annihilation operators over the vector space $\mathcal{H}$ of each
local Fermionic mode labeled by $x \in \mathcal{L}$. These operators satisfy
Canonical Anticommutation Relations (CAR)
\begin{equation}
\label{eq:CAR}
\{\phi_{{x}},\phi_{{y}}^\dag\}=\delta_{{x}{y}} I, \qquad \{\phi_{x},\phi_{x}\}= \{\phi^\dag_{{x}},\phi^\dag_{{y}}\}=0, 
\end{equation}
for all ${x},{y} \in \mathcal{L}$, where $\delta_{{x}{y}}$ is the
standard Kronecker delta over integers, $\{A,B\}\coloneqq AB + BA$ is
the anticommutator and $I$ is the identity operator. We consider now 
the occupation number operator
$N_{{x}} \coloneqq \phi_{{x}}^\dag \phi_{{x}}$ on the site $x$ and its
basis $\{\ket{0_{{x}}},\ket{1_{{x}}}\}$ of eigenvectors
\begin{align}
\begin{aligned}
&N_{{x}}\ket{0_{{x}}}=0, \quad  N_{{x}}\ket{1_{{x}}}=\ket{1_{{x}}}, \\
&\phi_{{x}}\ket{1_{{x}}}=\ket{0_{{x}}}, \quad \phi^\dag_{{x}}\ket{0_{{x}}}=\ket{1_{{x}}}, \quad
\phi_{{x}}\ket{0_{{x}}}=\phi^\dag_{{x}}\ket{1_{{x}}} =0.
\end{aligned}\label{eq:basis}
\end{align}
In the \emph{local Fermionic mode picture} the above vectors are
interpreted as zero and one particle states at position
${x}$, with $\phi_x,\phi_x^\dag$ the creation/annihilation operators of a
Fermionic excitation in the $x$-th site.

Let us consider the \textit{vacuum} with trivial multiplicity
${\Omega} \in \mathbb{C}$, i.e., the unique vector state which is
simultaneous eigenvector with eigenvalue 0 common to all $N_x$, namely
$N_{x}\ket{\Omega}=0 \hspace{10pt} \forall x\in\mathcal{L}$. The
vacuum state, corresponding to all sites of the lattice being
unoccupied, is annihilated by all the annihilation operators, i.e.,
$\phi_x \ket{\Omega} = 0$ for all $x \in \mathcal{L}$. By acting on
$\ket{\Omega}$ with arbitrary combinations of products of creation
operators $\phi^\dag_x$, the Fock space for a system of local
Fermionic modes given in \Eq~\eqref{eq:fock-der} may be identified
with the space

\begin{equation}
\label{eq:fock}
	\mathscr{F}=\operatorname{Span_{\mathbb{R}}}\left\{\prod_{{x}\in \mathcal{L}}(\phi_{{x}}^\dag)^{p_{{x}}}\ket{\Omega}, \quad p_{{x}}=0,1 \ \forall {x}\in\mathcal{L}\right\},
\end{equation}
where $p_x$ denotes the occupation number at the $x$-th site, and the product 
$\Pi_{x\in\mathcal L}$ is well defined provided that a total (immaterial) ordering is 
introduced on the denumerable set $\mathcal L$.

Using the creation/annihilation operators $\phi_x, \phi^\dag_x$, one
can define the Fermionic Pauli matrices
\begin{align}\label{eq:Fpauli}
	X_x = \frac{\phi_x+\phi_x^\dag}{\sqrt{2}}, && Y_x= \frac{-i(\phi_x-\phi^\dag_x)}{\sqrt{2}},
\end{align}
for each site $x$, whence
\begin{align}\label{eq:Fpauli-inv}
	\phi_x=\frac{(X_x - iY_x)}{\sqrt{2}}, && \phi_x^\dag=\frac{(X_x +i Y_x)}{\sqrt{2}}.
\end{align}
The Fermionic Pauli matrix $Z_x$ is defined as
\begin{align*}
	Z_x \coloneqq i Y_x X_x  = \frac{\phi_x^\dag \phi_x - \phi_x \phi^\dag_x}{2}.
\end{align*}
In the orthonormal basis of
\Eq~\eqref{eq:basis}, the matrix representation of $X_x,Y_x,Z_x$ is given by
\begin{align*}
	\begin{aligned}
		X_x =\begin{pmatrix} 0&1\\1&0\end{pmatrix}, \quad Y_x=\begin{pmatrix} 0&-i\\i&0\end{pmatrix},\quad Z_x=\begin{pmatrix} 1&0\\0&-1\end{pmatrix}.
	\end{aligned}
\end{align*}
This notation should not be confused with the usual qubits Pauli
matrices $\sigma^{(x)}_x,\sigma^{(y)}_x,\sigma^{(z)}_x$ that commute
on different sites of the lattice, unlike $X_x,Y_x,Z_x$, which anticommute.

\subsection{CAR and $\mathbb{Z}_2$-graded \textit{quasi-local} algebras}

In the framework of operator algebra and local quantum theory, a set
of Fermionic modes arranged on the lattice
$\mathcal{L} $ is  described by a $C^*$-algebra
over an infinite dimensional complex Hilbert space, called
\emph{CAR algebra}. For a general review of CAR algebras, see, e.g.,
\cite{bratteli2012operator,bratteli2013statistical,Derezinski2006}. In
this algebraic approach, a pair of self-adjoint anticommuting
operators $(\xi,\eta)$ is assigned to each site $x\in\mathcal{L}$,
such as the Fermionic Pauli Matrices $(X_x,Y_x)$ defined in
\Eq~\eqref{eq:Fpauli}. These operators are then promoted to be the
generators of the self-adjoint algebra of operators localized at
$x$. Each local algebra is isomorphic to the algebra of $2 \times 2$ complex
matrices containing the identity $I_x$, in such a way that each
operator $O_x$ on $x$ is a linear combinations of the
identity and the Fermionic Pauli matrices \eqref{eq:Fpauli}
\begin{equation*}
	O_x= a I+ b X_x+ c Y_x+ d Z_x, \qquad a,b,c,d \in \mathbb{C}.
\end{equation*}	

Unlike the case of qubits, Fermionic
systems obey the \emph{parity superselection}
rule. The Fock space \eqref{eq:fock} can be decomposed as follows
\begin{equation}
	\mathscr{F}=\mathscr{F}_e\oplus\mathscr{F}_o,
	\label{eq:Fermgrad}
\end{equation}
where $\mathscr{F}_{e/o}$ are the eigenspaces of the parity operator
\begin{equation}
	\label{eq:parity}
	P=\frac{1}{2}\{I+ \prod_{{x}\in\mathcal{L}}(\phi_{{x}}\phi^\dag_{{x}}-\phi_{{x}}^\dag\phi_{{x}})\}
\end{equation}
with eigenvalues $0/1$, respectively.  The subspaces
$\mathscr{F}_{e/o}$ contain vectors corresponding to an even/odd total occupation
number, and parity superselection amounts to requiring that superpositions of  
vectors with different occupation number parity are forbidden. 
More precisely, the states of a Fermionic chain are density matrices with a 
well-defined parity:
\begin{theorem}[Parity Superselection]
	\label{thm:Parity}
	Density matrices representing Fermionic states commute with the total parity operator $P$ given in \Eq~\eqref{eq:parity}.
\end{theorem}

Consequently, linear maps representing transformations are bound to 
respect the parity superselction of states. As shown in Ref.~\cite{doi:10.1142/S0217751X14300257}, this implies that the operators corresponding to Kraus operators of allowed maps have a well-defined parity, thus form two subspaces (even or odd) of the CAR algebra. For instance, both $X_x$ and $Y_x$ are odd operators 
since they anticommute with $P$, whereas their product $Z_x$ commutes with $P$ and 
is an even operator.  This turns the CAR algebra into a $\mathbb{Z}_2$-graded
algebra~\cite{bratteli2012operator}.

\begin{definition}[$\mathbb{Z}_2$\textit{-graded algebra}]
A $\mathbb{Z}_2$-graded algebra $\mathsf{A}$ is a sequence
\begin{equation*}
\mathsf{A}=(\mathsf{A}^0,\mathsf{A}^1),
\end{equation*} 
where $\mathsf{A}^{0,1}$ are respectively the modules of even and odd operators. $\mathsf{A}$ is equipped with the following operations:
\begin{enumerate}
\item Graded sum: The sum is defined only between elements of $\mathsf{A}$ with the same parity, i.e.,
\begin{align*}
	O_1,O_2\in \mathsf{A}^p, \qquad O_1+O_2\in \mathsf{A}^p.
\end{align*}
\item Graded product: For every two operators we have \begin{align*}
	O_1\in\mathsf{A}^p,O_2\in\mathsf{A}^q, \qquad O_1O_2\in\mathsf{A}^{p\oplus q},
\end{align*}
where $\oplus$ denotes sum modulo 2.
\item Graded tensor product: For every two $\mathbb{Z}_2$-\textit{graded algebras} $\mathsf{A}$ and $\mathsf{B}$, we can define their graded tensor product $\mathsf{A}\boxtimes\mathsf{B}$ as
\begin{equation*}
\begin{split}
&\mathsf{A}\boxtimes\mathsf{B}=((\mathsf{A}\boxtimes\mathsf{B})^0,(\mathsf{A}\boxtimes\mathsf{B})^1), \\
&(\mathsf{A}\boxtimes\mathsf{B})^i=\bigoplus_{\substack{p,q=0,1 \\ p\oplus q=i}} \mathsf{A}^p\gt\mathsf{B}^q,
\end{split}
\end{equation*}
and denoting by $g(T)$ the grade of $T$, i.e.~$g(T)=0$ if $T\in\mathsf A^0$ and 
$g(T)=1$ if $T\in\mathsf A^1$, we have
\begin{align*}
(T\boxtimes I)(I\boxtimes S)=(-1)^{g(S)g(T)}(I\boxtimes S)(T\boxtimes I),
\end{align*}
where $I$ denotes the identity operator with $g(I)=0$.
\end{enumerate}
\end{definition}
Notice that, by the above definition, $g(A\boxtimes B)=g(A)\oplus g(B)$.
It is also convenient to introduce the following notion.
\begin{definition}[Graded commutator]
  For every $O_1,O_2\in\mathsf{A}$, we define
  the graded commutator as:
\begin{equation*}
\gc{O_1}{O_2}=O_1O_2-(-1)^{g(O_1)g(O_2)}O_2O_1.
\end{equation*}
\end{definition}
The graded commutator reduces to the anticommutator if $O_1,O_2$ are both odd, 
and to the commutator otherwise.

We can now define the local Fermionic algebra on $\mathcal L$ as the algebra of local operators, i.e.~operators acting non-trivially on 
finitely many sites of the lattice $\mathcal L$, whose completion
in the operator norm provides a graded algebra called \textit{quasi-local}
algebra~\cite{bratteli2012operator}. More precisely, upon denoting the single 
$x$-th cell algebra generated by the pair $(X_x,Y_x)$ as $\mathsf{A}_x$, 
if we consider a finite cardinality subset $\Lambda\subset \mathbb{Z}$ of the
chain, we can define the \textit{local} algebra at $\Lambda$ as:
\begin{align*}
&\mathsf{A}_{\Lambda_1}\coloneqq \gt_{x\in\Lambda_1} \mathsf{A}_x, \qquad O_{\Lambda_i}\in\mathsf{A}_{\Lambda_i},
\end{align*}
and then the \emph{local algebra} on $\mathcal L$ will be characterised by the product
\begin{align*}
&O_{\Lambda_1}O_{\Lambda_2}\coloneqq(O_{\Lambda_1}\gt I_{\Lambda_2\setminus\Lambda_1})(I_{\Lambda_1\setminus\Lambda_2}\gt O_{\Lambda_2}).
\end{align*}
Since $\mathsf{A}_x$, as well as $\mathsf{A}_{\Lambda}$, is a normed
space with the standard operator norm, we can consider the closure of
the local algebra in this norm, providing the graded
\textit{quasi-local} algebra $\mathsf{A}(\mathbb{Z})$ over
$\mathbb{Z}$. This is the algebra of operators that can be
approximated with arbitrary precision with local operators.

\section{One Dimensional Fermionic Cellular Automata}
The notion of a Fermionic Cellular Automaton is the Fermionic counterpart of 
that of Quantum Cellular Automaton defined in \cite{schumacher2004reversible}, 
which is an automorphism of $\mathsf{A}(\mathbb{Z})$. Preliminarly, we 
introduce a special class of FCA, called shifts, which implement lattice
translations.

\begin{definition}[Shift $\tau_x$]\label{def:shift}
  A shift $\tau_x$ is the automorphism of $\mathsf A(\mathbb Z)$ such that 
\begin{align*}
\tau_x(X_y) &\coloneqq X_{x+y},\qquad \tau_x(Y_y)\coloneqq Y_{x+y}.
\end{align*}
\end{definition}
Fermionic Cellular Automata are then defined as follows:
 \begin{definition}\label{def:FCA} 
   A $\emph{Fermionic Cellular Automaton}$ (FCA)
  with finite neighbourhood scheme $\mathcal{N} \subset\mathbb{Z}$ is an
  automorphism
   $\tT:\mathsf{A}(\mathbb{Z}) \longrightarrow
   \mathsf{A}(\mathbb{Z})$ of the \textit{quasi}-local $\mathbb{Z}_2$-graded
   algebra such that
   \begin{enumerate}
   \item  $\tT(\mathsf{A}(\Lambda))\subset
   \mathsf{A}(\Lambda+\mathcal{N}) \quad
   \forall \Lambda\subset \mathbb{Z}$ (locality),
   \item $\tT\circ\tau_x=\tau_x\circ\tT \quad$ 
   $\forall x\in \mathbb{Z}$ (homogeneity),
 \end{enumerate}
 where $\Lambda+\mathcal{N} \coloneqq \{ y+x | y \in
\Lambda, x \in \mathcal{N}  \}$.
   The homomorphism $\tT_0:\mathsf{A}_0\longrightarrow
   \mathsf{A}(\mathcal{N})$ given by the
   restriction of the FCA to the $0$ cell is called the
   $\emph{local transition rule}$.
\end{definition}

Defining a FCA $\tT$ as a $*$-automorphism over
$\Aqloc=(\m A_0,\m A_1)$ implies that each $\tT$ preserves the parity
of every element in $\Aqloc$. In particular,
$\acomm{\tT(\xi)}{\tT(\eta)}=0$ for all pair of odd generators
$(\xi,\eta)$ satisfying $\acomm{\xi}{\eta}=0$. Moreover, since the
full-graded matrix algebra is generated by two anticommuting odd
operators, it is sufficient to specify the action of $\tT$ over those
operators to fully characterize $\tT$ over $\Aqloc$.

Finally, the following theorem specifies the relation between the local and the global rules of the FCA, i.e., respectively the *-homomorphisms describing the global time evolution of $\tT$ and its restriction to the single cell algebra \cite{schumacher2004reversible}.
 
\begin{theorem} The following statements hold:
	\label{thm:locglo}
\hspace*{1pt}
\begin{enumerate}
	\item The global *-homomorphism $\tT$ is uniquely determined
	by the local transition rule $\tT_0$.\\
	\item A *-homomorphism $\tT_0 : \mathsf{A}_{0} \longrightarrow \mathsf{A}_\mathcal{N}$ is the transition
	rule of a FCA if and only if, for all $x \in
	\mathbb{Z}$ such that $\mathcal{N}\cap(\mathcal{N}+x)\neq \emptyset$, the algebras $\tT_{0}(\mathsf{A}_{0})$ and $\tau_x \tT_{0}\tau_x^{-1}(\mathsf{A}_{x})$ graded-commute element-wise.
\end{enumerate}
The relation between the local rule and the global evolution is given by
\begin{equation}
   \tT(\boxtimes_{x\in\Lambda} \m A_x)=\prod_{x\in \Lambda}\tT_x(\m A_x),
\end{equation}
where $\tT_x\coloneqq \tau_x\tT_0\tau_x^{-1}$, and the same (immaterial) ordering of 
$\Lambda$ is implicit in both products.
\end{theorem}

As illustrative examples of FCA, we consider the family of conjugation 
automorphisms and the so-called \textit{Majorana}-shift. On the one
hand, every conjugation maps a single-cell operator into another one
by conjugation with a parity preserving unitary matrix. We remind that the group of 
reversible maps on a single Fermionic mode is
\begin{equation}
G\cong\mathbb U(1)\times \mathbb Z_2,
\end{equation}
whose elements are given by 
\begin{align}\label{eq:unit}
U(\theta,n)\coloneqq e^{i\theta Z}X^n, \quad n=0,1. 
\end{align}
\begin{definition}[Conjugation $\mathscr{U}_x$]
	\label{def:Conjugation}
A conjugation automaton is given by
	\begin{align*}
		&\mathscr{U}_x^{(\theta,n)}: \Aqloc \rightarrow \Aqloc,\\
		&\mathscr{U}_x^{(\theta,n)}(\A_x) \coloneqq U(\theta,n)^\dag \A_x U(\theta,n).
	\end{align*}
\end{definition}
In the remainder, we will omit the parameters $\theta,n$ when they are clearly identified 
by the context. On the other hand, the Majorana shift is defined as follows.
\begin{definition}[\textit{Majorana}-shift $\sigma_{\pm}$]
	\label{def:Majoranashift}
	The \textit{Majorana}-shifts
	\begin{equation}
	\label{eq:Majoshift}
		\sigma_{\pm}:\A(\mathbb{Z})\rightarrow \A(\mathbb{Z})
	\end{equation}
	are the (unique) homogeneous FCAs such that
	\begin{equation}
		\label{eq:Mshift}
		\begin{aligned}
                  \sigma_\pm&: (X_{x},Y_{x}) \mapsto (Y_{x},X_{x\pm1}), 
		\end{aligned}
	\end{equation}	
	along with their inverses
	\begin{align}
        \sigma^{-1}_\pm&: (X_{x},Y_{x}) \mapsto (Y_{x\mp1},X_{x}).
	\end{align}
\end{definition}

\begin{remark}[Majorana modes]
Up to now we have considered chains where at each site $x$ lies a
Fermion, whose algebra $\mathsf{A}_x$ is generated by
the pair $(X_x,Y_x)$. It is possible to decouple each Fermionic
mode, say the $i$-th one, into
two \emph{Majorana modes} as follows:
\begin{align}
\xi_{2x}\coloneqq X_x,\quad\xi_{2x+1}\coloneqq Y_x.\label{eq:mamod}
\end{align}
        Accordingly, we can think of a chain of ``complex'' $n$
        Fermionic modes as a chain of $2m$ Majorana modes, with two
        Majorana modes for each site $x$. Based on
        Eqs.~\eqref{eq:mamod} and~\eqref{eq:Fpauli-inv} one can
        identify Majorana modes with the generators of the local
        algebra, thus justifying the ``Majorna-shifts'' nomenclature
        for the automata~\eqref{eq:Majoshift}. Indeed, $\sigma_+$ (and its inverse)  
        move the Majorana modes to the right (left)~\cite{PhysRevB.99.085115}. 
		On the other hand, $\sigma_-$ (and its inverse) have a similar behaviour, 
		provided that we first exchange $\xi_{2x}$ and $\xi_{2x+1}$.
\end{remark}

\subsection{Index for FCA}
We now review the concept of index for a one-dimensional Fermionic
QCA. The theory was presented in \cite{Werner-index} for qudits
quantum walks and QCA and it was then extended to the Fermionic case
in \cite{PhysRevB.99.085115}. For a review of the index theory, see
e.g. \cite{Farrelly2020reviewofquantum}.  Naively speaking, the index
summarises all local invariants of a QCA---i.e.~quantities that can be evaluated
on finitely many sites and are invariant along the lattice---
We begin reviewing the concept of \emph{support algebras}
\cite{schumacher2004reversible}\cite{PhysRevB.99.085115}.

\begin{definition}[\textbf{Support Algebra}]
  Let $\mathsf{A},\mathsf{B}_{L},\mathsf{B}_R$ be
  ($\mathbb{Z}_2$-graded) algebras, with
  $\mathsf{A}\subset \m B_{L}\boxtimes \m B_{R}$. We denote the support of
  $\mathsf{A}$ on $\mathsf{B}_L$ (or on $\mathsf{B}_R$) by
  $\mathbf{S}(\mathsf{A},\mathsf{B}_L)$ (or $\mathbf{S}(\mathsf{A},\mathsf{B}_R)$). 
  This is defined to be the smallest ($\mathbb Z_2$-graded) 
  $C^*$-subalgebra $\mathsf{S}\subset\mathsf{B}_L$ 
  such that $\mathsf{A}\subset\mathsf{S}\otimes\mathsf{B}_R$ (or the smallest 
  ($\mathbb Z_2$-graded) 
  $C^*$-subalgebra $\mathsf{S}\subset\mathsf{B}_R$ 
  such that $\mathsf{A}\subset\mathsf{B}_L\otimes\mathsf{S}$). 
\end{definition}
One can very easily prove the following result that provides a characterisation of 
the support algebras~\cite{Werner-index,PhysRevB.99.085115}.
\begin{lemma}
  Let $\mathsf{A},\mathsf{B}_{L},\mathsf{B}_R$ be
  ($\mathbb{Z}_2$-graded) algebras, with
  $\mathsf{A}\subset B_{L}\boxtimes B_{R}$.
  , for any $a\in\mathsf{A}$, we can decompose it as
\begin{equation*}
a=\sum_{i} b_{L,i}\boxtimes b_{R,i},
\end{equation*}
where $\{b_{L,i}\}\subseteq\mathsf{B}_L$ and $\{b_{R,i}\}\subseteq\mathsf{B}_R$ are  linearly independent sets. Then $\mathbf{S}(\mathsf{A},\mathsf{B}_L)$ and $\mathbf{S}(\mathsf{A},\mathsf{B}_R)$ are the ($\mathbb Z_2$-graded) algebras generated by these two sets, respectively.
\end{lemma}
Clearly, given $\mathsf{A}\subset\m B_{L}\boxtimes \m B_{R}$, one has
\begin{align}
\mathsf A\subseteq\mathbf{S}(\mathsf{A},\mathsf{B}_L)\boxtimes \mathbf{S}(\mathsf{A},\mathsf{B}_R)\label{eq:subsind}.
\end{align}

One of the most relevant properties of Support Algebras, that is crucial to 
classify one-dimensional FCA, is contained in the following Lemma~\cite{PhysRevB.99.085115}. This is a straightforward generalisation to the
graded case of Lemma~8 in Ref.~\cite{Werner-index}.
\begin{lemma} \label{thm:suppcomm} Consider the ($\mathbb{Z}_2$-graded)
  algebras $\mathsf{B}_L,\mathsf{B}_C,\mathsf{B}_R$, and let
  $\mathsf{A}_{LC},\mathsf{A}_{CR}$ be subalgebras of
  $\mathsf{B}_L\boxtimes\mathsf{B}_C\boxtimes\mathsf{B}_R$ satisfying
\begin{equation*}
  \mathsf{A}_{LC}\subseteq \mathsf{B}_L\boxtimes\mathsf{B}_C,\qquad
  \mathsf{A}_{CR}\subseteq \mathsf{B}_C\boxtimes\mathsf{B}_R.
\end{equation*}
Then $\gc{\mathsf{A}_{LC}\boxtimes I_R}{I_L\boxtimes\mathsf{A}_{CR}}=0$ implies
\begin{equation*}
\gc{\mathbf{S}(\mathsf{A}_{LC},\mathsf{B}_C)}{\mathbf{S}(\mathsf{A}_{CR},\mathsf{B}_C)}=0.
\end{equation*}
\end{lemma}
We can now define the so-called left and right support algebras for a QCA (FCA) 
$\tT$ as
\begin{equation}\label{eq:left-right}
\begin{split}
  \mathsf{L}_{2x} \coloneqq \mathbf{S}(\tT(\mathsf{A}_{2x}\boxtimes\mathsf{A}_{2x+1}),\A_{2x-1}\boxtimes\A_{2x}),\\
  \mathsf{R}_{2x+1} \coloneqq
  \mathbf{S}(\tT(\mathsf{A}_{2x}\boxtimes\mathsf{A}_{2x+1}),\A_{2x+1}\boxtimes\A_{2x+2}).
\end{split}
\end{equation}
The core result on which index theory hinges is the following lemma, that was proved 
for the qudit case in~Ref.\cite{Werner-index} and generalised to the Fermionic case in Ref.~\cite{PhysRevB.99.085115}, consists in the following stronger version of the inclusion in~\eqref{eq:subsind}. 
\begin{lemma}
Let $\tT$ be a QCA (FCA). Then 
\begin{equation}\label{eq:suppaut}
\mathsf{L}_{2x}\gt\mathsf{R}_{2x+1}=\tT(\A_{2x}\gt\A_{2x+1}).
\end{equation}
\end{lemma}
Now, we can define the index as 
follows.
\begin{definition}[Index]
  The index of a QCA (FCA) having global evolution $\tT$ is
\begin{equation}
  \operatorname{ind}[\tT] \coloneqq \sqrt{\frac{\operatorname{dim}[\mathsf{L}_{2x}]}{\operatorname{dim}[\A_{2x}]}}=\sqrt{\frac{\operatorname{dim}[\mathsf{A}_{2x+1}]}{\operatorname{dim}[\m R_{2x+1}]}},
\label{eq:ind}
\end{equation}
with the algebras $\mathsf{L}_{2x}$ and $\mathsf{R}_{2x+1}$ defined in
Eq.~\eqref{eq:left-right}, and where the last equality follows from
\Eq\eqref{eq:suppaut}.
\end{definition}

Via results on classification of (semi)simple algebras it is possible
to find which values the index can take in principle. As stated in
\cite{PhysRevB.99.085115}, typical classification results hold also in
the graded case \cite{PhysRevB.99.085115}. It was shown that a generic
support algebra $\m S_x$ is a $\textit{simple}$ $\mathbb{Z}_2$ graded
algebra, i.e. a $\textit{semisimple}$ $\mathbb{Z}_2$ graded algebra
with trivial graded center. Therefore, a $\mathbb{Z}_2$ graded version
of the Wedderburn's Theorem holds in this case. Based on this, and
employing the notation of \cite{PhysRevB.99.085115}, the following
Lemma holds:

\begin{lemma}\label{lem:alg}
  Let $p,q\in \mathbb{N}$ the dimensions of the even and odd sectors
  of the $\mathbb{Z}_2$-graded vector space $\mathbb{C}^{p\lvert
    q}$. A semisimple graded algebra over $\mathbb{C}^{p\lvert q}$
  with trivial center is isomorphic to either:
\begin{enumerate}
\item $\M{p}{q}$, the graded algebra of matrices acting on the graded
  space $\mathbb{C}^{p\lvert q}$. As complex vector space, it has
  dimension $(p+q)^2$
\item $\mathit{C}\ell_1(p\lvert q)$ the $\mathbb{Z}_2$-graded algebra
  of $(p+q)\times(p+q)$ block-diagonal or anti-diagonal matrices 
  with entries in the one dimensional complex Clifford algebra
  $\textit{C}\ell_1(\mathbb{C})$. Specifically, an element of
  $\Cl{p}{q}$ is of the form $ M_1\gt I+M_2\gt\gamma$ where
  $M_i\in\M{p}{q}$, while $\gamma$ is a Majorana mode. As complex
  vector space, it has dimension $2(p+q)^2$.
\end{enumerate}
\end{lemma}

Therefore, $\operatorname{dim}(\m {S}_{x})= m(p+q)^2$, where either
$m=1$ if $\m S_{x} \simeq \M{p}{q}$ or $m=2$ if $\m {S}_{x} \simeq
\mathit{C}\ell_1(p\lvert q)$. Since both $\operatorname{dim}(\m
{S}_{x})$ and $d^2 \coloneqq \operatorname{dim}[\A^{x}]$ do not depend
on the choice of the site $x$, then the index \eqref{eq:ind} is also
site-independent, and hence can take values
\begin{equation}
\operatorname{ind}[\tT]=\sqrt{m}\frac{p+q}{d}.
\label{eq:indform}
\end{equation}

\subsection{Local implementability and equivalence relations}

As discussed in Refs.~\cite{Werner-index,PhysRevB.99.085115}, the
index of a (Fermionic) Quantum Cellular Automaton determines its
\textit{circuital implementability}, namely whether it can be
represented by finite sequences of blocks of unitaries
(\textit{gates}). To elaborate on this property, it is convenient to
introduce the notion of \emph{finite depth Fermionic circuit}, which is the
prototype of implementable quantum circuit.

\begin{definition}[Finite Depth Fermionic Circuit (FDFC)]
	\label{def:fdqc} 
	A Finite Depth Fermionic Circuit is an automorphism $\tF$ of the quasi-local graded algebra $\A(\mathbb{Z})$, such that there exist a positive integer $D<\infty$ (the \emph{depth} of $\tF$) and, for every $t=1,2,\ldots D$, 
	a bounded size partition $\{\mu_{i,t}\}_{i\in\mathbb Z}$ of the lattice 
	$\Lambda$ and for every $i\in\mathbb Z$ and $1\leq t\leq D$ a unitary 
	$U_{i,t}\in\mathsf A_{\mu_{i,t}}$, so that 
for every local operator  $T\in\mathsf A_\Lambda$ 
\begin{equation}
\label{eq:fdqc}
\begin{aligned}
	&\tF(T) = F_\Lambda T F_\Lambda^\dag,\\
	&F_\Lambda\coloneqq\prod_{s=0}^{D-1}\bigboxtimes_{i\in J_{D-s}} U_{i,D-s},\\
	&\Lambda_0\coloneqq\Lambda, \quad J_t\coloneqq\{i\in\mathbb Z\mid\mu_{i,t}\cap\Lambda_{t-1}\neq\emptyset\},\\
	&\Lambda_t\coloneqq\bigcup_{i\in J_t}\mu_{i,t}.
\end{aligned}
\end{equation}
\end{definition}
In words, a FDFC is a sequence of layers of unitary transformations, where each 
unitary in a given layer acts over a finite set of cells in a partition of the lattice 
that generally changes in each layer.
\begin{remark}
  Clearly, any FDFC \eqref{eq:fdqc} is an example of FCA. Indeed,
  taking into account the local rule of the FDFC as in Definition
  \ref{def:FCA}, the condition $|\mu_{i,t}|<k$ for every pair $i,t$ 
  ensures that there
  exists a finite neighbourhood scheme
  $\mathcal{N}_{\tF} \subset \mathbb{Z}$ such that
\begin{align*}
  \tF(\mathsf{A}(\Lambda))\subset \mathsf{A}(\Lambda+\mathcal{N}_{\tF}), \qquad \lvert \mathcal{N}_{\tF}\rvert <\infty.
\end{align*}
Thus, $\tF$ has a bounded neighbourhood scheme.  On
the other hand, it is not required that $\mathcal{N}_{\tF}$ to be a
nearest neighbourhood scheme, i.e., generally
$\mathcal{N}_{\tF}\neq \{-1,0,1\}$.
\end{remark}

A relevant class of FDFC is that of Margolus Partition schemes, i.e.~FDFC with depth $D=2$ and nearest neighbourhood scheme.

\begin{definition}[Margolus Partitioned Scheme (MS)]
\label{def:Marg}
A Margolus Partitioned Scheme $\tM$ is an FDFC of depth $D=2$, with 
the two layers of unitaries over nearest neighbours $\{2x,2x+1\}$ and
$\{2x+1,2x\}$ respectively, namely, for $T\in\mathsf A_\Lambda$, defining 
$J_\Lambda\coloneqq\{x\mid 2x\in\Lambda\vee 2x+1\in\Lambda\}$ and 
$K_\Lambda\coloneqq\{x\mid 2x+2\in \Lambda\vee 2x+1\in \Lambda\}$
\begin{equation}
	\label{eq:Margolus}
\begin{aligned}
 \tM \coloneqq &\tM_2\circ\tM_1\\
  &\tM_1(T) \coloneqq
  \left(\bigboxtimes_{x\in J_\Lambda}M^{(1)}_{2x}\right)T\left(\bigboxtimes_{x'\in J_\Lambda}M^{(1)}_{2x'}\right),\\
&\tM_2(T) 
  \coloneqq \left(\bigboxtimes_{x\in K_\Lambda} M^{(2)}_{2x+1}\right)T\left(\bigboxtimes_{x'\in K_\Lambda} M^{(2)}_{2x'+1}\right),
\end{aligned}
\end{equation}
with the unitaries $M^{(i)}_y$ acting over $\{y,y+1\}$.
\end{definition}

Based on the definitions of FDFC and MS and inspired by
Ref.~\cite{Werner-index}, we can formulate a general definition of
implementability for FCA in terms of FDFC and MS, respectively.

\begin{definition}[\Fimpl]
	\label{def:implem}
	A FCA $\tT: \Aqloc \rightarrow \Aqloc$ is said to be
        $\tF$-implementable if there exists a FDFC
        \eqref{eq:fdqc} such that $
		\tT=\tF$
\end{definition}

In the special case of FDFC replaced by MS, this definition reduces to
the \textit{local implementability} already defined in
\cite{Werner-index} for ungraded QCA. In this paper, we can
reformulate that notion for FCA as follows.

\begin{definition}[\Mimpl]
	\label{def:implemloc}
		A FCA $\tT: \Aqloc \rightarrow \Aqloc$ is said to be
        $\tM$-implementable if there exists a Margolus Partition Scheme 
        \eqref{eq:Margolus} such that $\tT=\tM$.
\end{definition}

Based on the above notions we can introduce \textit{equivalence}
classes of Fermionic Cellular Automata. The original
definition of equivalence was provided in \cite{Werner-index} for the
ungraded case, in order to classify one-dimensional QCA. Naively
speaking, two different Quantum Cellular Automata are said to be
equivalent whenever they can be deformed into each other through a 
Finite Depth Quantum Circuit (FDQC).
Hence, it was proved that all the ungraded QCA
with the same index are equivalent in this sense. Here we
introduce an analogous definition in the Fermionic scenario, where one considers 
FDFC~\eqref{def:fdqc}, instead of FDQC. We denote this criterion by
$\Fequiv$.
%
%
%
%
\begin{definition}[\Fequiv]
\label{def:equiv}
Two Fermionic Cellular Automata $\tT$ and $\tS$ are
\textit{$\tF$-equivalent} if there exists a FDFC $\tF$ such that
	\begin{equation}
	\label{eq:ancillarem}
		\tT = \tF \circ \tS.
	\end{equation}
\end{definition}

Finally, the weaker notion  of \textit{stable equivalence} of Ref.~\cite{PhysRevB.99.085115} can be introduced, which is similar to 
$\Fequiv$, but it also allows for the addition of auxiliary cells.
\begin{definition}[Stable Equivalence]
	\label{def:stableeq}
	Two Fermionic Cellular Automata $\tT$ and $\tS$ are
        \textit{stably equivalent} if:
	\begin{equation}
	\label{eq:stabeq}
		\tT\gt \tI = \tilde{\tF} \circ(\tS\gt \tI),
	\end{equation}
        where $\tI$ is the identity automorphism over a suitable set
        of ancillary systems, and $\tilde{\tF}$ is a FDFC. 
\end{definition}
In Ref.~\cite{PhysRevB.99.085115} the following Lemma was proven, generalising the analogous result of Ref.~\cite{Werner-index} to the graded case.
\begin{lemma}\label{lem:Mequiv}
Let $\tT$ and $\tS$ be two stably equivalent FCA according to Definition \ref{def:stableeq}. Upon enlarging each original cell with an ancillary 
copy, and suitably regrouping the enlarged cells into supercells to match the 
dimension of the cells of $\tT$ and $\tS$, the two FCAs are stably 
$\tM$-equivalent through a MS $\tM$, i.e.,
\begin{align}\label{eq:Mequiv}
\tT\gt\tI=\tM_2\circ\tM_1\circ(\tS\gt\tI),
\end{align}
with $\tM_1\circ\tM_2\coloneqq \tM $ an MS over such supercells
according to Definition \ref{def:Marg}.
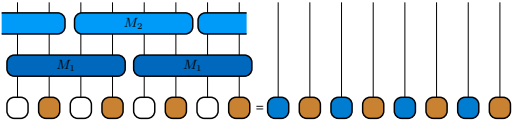
\begin{figure}[h!]
\resizebox{0.8\hsize}{!}{
\input{Ancilla.tikz}=\input{figures/Ancillamod.tikz}}
\caption{Graphical representation of \Mimpl of index one FCA. The
  white cells represent the cells of our register, the orange ones
  represent the ancillae, and the blue ones represent the cells
  evolved through $\tT$ with $\ind{\tT}=1$. }
\label{fig:Ancillaaction}
\end{figure}
\end{lemma}
In the remainder, when referring to stable equivalence, we will use the wording
\emph{cell} to denote the enlarged cell made of the original cell and its ancillary copy.

\subsection{Properties of the Fermionic index}

Stable equivalence was exploited in \cite{PhysRevB.99.085115} to
generalize to FCA the properties of the index derived in
\cite{Werner-index} for QCA. The result of
Ref.~\cite{PhysRevB.99.085115} is reviewed in the following Lemma:

\begin{lemma}[Index properties for FCA]\label{lem:prop}
The following properties hold:\\
\begin{itemize}
\item If $\tT$ is a FDFC, then $\ind{\tT}=1$.
\item Given two distinct FCA $\tT$ and $\tT'$,
\begin{equation}\label{eq:indmult}
\begin{split}
&\ind{\tT\tT'}=\ind{\tT}\ind{\tT'}\\
&\ind{\tT\gt \tT'}=\ind{\tT}\ind{\tT'}.
\end{split}
\end{equation}
\item If $\ind{\tT}=\ind{\tT'}$, then $\tT$ and $\tT'$ are
  \emph{stably equivalent} according to Definition
  \ref{def:stableeq}.
\end{itemize}
\end{lemma}

As shown in \cite{PhysRevB.99.085115}, combining stable
equivalence~\Eq\eqref{eq:stabeq} and the last item of
Lemma~\ref{lem:prop} it follows that every FCA with index one
is $\mathcal M$-implementable in the enlarged space of the lattice tensorized with
ancillae, i.e.
\begin{lemma}
If \ind{\tT}=1 then there exists a MS $\tM$ such that $\tT\gt\tI=\tM$.
\end{lemma}
\begin{proof}
  It is easy to see that the identity has index one. Thus, the last
  item of the previous lemma tells us that any index one FCA is
  stably equivalent to the identity. Considering the condition of stable 
  $\mathcal M$-equivalence in
  \Eq\eqref{eq:Mequiv} we have:
\begin{align}
\tT\gt\tI=\tM_2\circ\tM_1\circ(\tI\gt\tI)=\tM.
\end{align}
\end{proof}

\section{Stable Equivalence and Ancilla Removal}
\label{sec:stabequiv}
When we consider \Eqs\eqref{eq:stabeq},\eqref{eq:Mequiv}, the whole
circuit $\tilde{\tF}$ acts trivially over the ancillary systems, but
the single gates ($M_1$ at first and $M_2$ later, in Figure
\ref{fig:Ancillaaction}) composing the circuit $\tilde{\tF}$ do not.
It would be desirable to replace the property of stable
equivalence~\Eq\eqref{eq:stabeq} with our notion of
$\mathcal{F}$-equivalence~\Eq\eqref{eq:ancillarem}, which does not
involve ancillae, and thus refer to a more intrinsic notion of local
implementability of the evolution.

In Ref.~\cite{Freedman-Haah-Hastings-group-structure}, the authors
introduced a method called \textit{ancilla removal}, which consists of
exploiting physical cells to play the role of the ancillary systems,
thus avoiding the addition of auxiliary systems. Our goal is to take
advantage of ancilla removal to show that two index one FCA are
always connected through an FDFC without referring to further
ancillary systems, i.e., they are equivalent according to Definition
\ref{def:equiv}. We observe that the local implementation of
a FCA based on a FDFC on the physical cells only can be more involved
than the simple MS needed with ancillae. 
\begin{figure}
\begin{tikzpicture}
\node at (-4.5,0) {\input{Gateancilla.tikz}};
\node at (-2.5,0) {=};
\node at (-0.75,0) {\input{Gateactionancilla.tikz}};
\end{tikzpicture}
\caption{Graphical representation of the action of the gates $M_1$ in
  \Eq\eqref{eq:Mequiv}. Here the purple and the green cells represent
  evolved sites and ancillary states respectively.}
\end{figure}
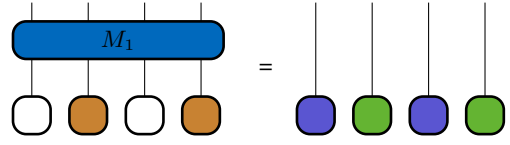

The possibility of performing ancilla removal with a bounded number of
ancillae for every physical cell was proved in
\cite{Freedman-Haah-Hastings-group-structure}. Before removing
ancillae, in the following proposition, we prove that only a single
ancilla per cell is needed for an index one FCA to achieve the stable
equivalence. The proof follows the same approach used in the appendix
of \cite{PhysRevB.99.085115}, which in turn refines the argument of
~\cite{Werner-index} for the ungraded case.
\begin{proposition}
  Two index one FCA $\tT$ and $\tS$ are related through the following
  stable equivalence relation:
\begin{equation}\label{eq:stable}
	\tilde{\tF} \circ(\tT\gt \tI)=\tS\gt \tI,
\end{equation}
The equivalence can be obtained through a MS.
\end{proposition}
\begin{proof}
  In the proof of Theorem 9 given in \cite{Werner-index} it was proven
  that two QCA with isomorphic support algebras are related through an
  MS. Thus the proof reduces to show the isomorphism of the support
  algebras $\m L_{2x}$ of $\tT$ and $\tS$ (and the same for $\m R_{2x+1}$). 
  In the non-graded case, this is
  trivially true for every index 1 QCA. Indeed :
\begin{align*}
\ind{\tT}=\sqrt{\frac{\dim{\m {L}_{2x}}}{\dim{\mathsf{A}_{2x}}}}=\sqrt{\frac{\dim{\m A_{2x+1}}}{\m {R}_{2x+1}}}=1,\\
\dim{\m {L}_{2x}}=\dim{\m R_{2x+1}}=\dim{\mathsf{A}_{2x}}=\dim{\m A_{2x+1}}.
\end{align*}
A key consideration in the proof is that full matrix algebras 
with the same dimension are isomorphic. This allows one to conclude that all 
support algebras are isomorphic to $\m A_x$, and thus all index one QCA are 
implementable through an MS: in a sketchy recap, the first layer arranges the 
counter-image of $\m L_{2x}$ in the cell $2x$ and that of $\m R_{2x+1}$ in the cell $2x+1$, and the second layer distributes $\m L_{2x}$ and $\m R_{2x-1}$ in the two 
cells $2x-1$ and $2x$. Since two index 1 QCA are implementable through a MS, 
$\tS\tT^{-1}$ is $\tF$-implementable, too, and thus $\tS=\tF\tT$, 
without the necessity of ancillary systems.
On the other hand, in the graded case the single cell algebra generated by $n$ 
Fermionic modes is a
full-matrix graded algebra $\M{n}{n}$. The support algebras $\m L_{2x}$ and 
$\m R_{2x+1}$ can be either Clifford algebras $\Cl{p}{q}$ or full matrix graded 
algebra $\M{p}{q}$. However, the Clifford case is ruled out by imposing that
the dimension of the support algebra is equal to the dimension of the
full matrix. Indeed, we have:
\begin{align*}
\dim{\M{p}{q}}=(p+q)^2=\dim{\M{n}{n}}=(2n)^2.
\end{align*}
Moreover---contrarily to the non-graded case---two graded algebras with equal 
dimension are not necessarily isomorphic
to each other.  In particular there may exist two index 1 QCA $\tT$
and $\tS$ having supports isomorphic to $\M{p}{q}$ and $\M{p'}{q'}$
such that:
\begin{equation}\label{eq:dimalgebra}
\begin{split}
&(p+q)^2=(p'+q')^2=(2n)^2,\\
&p\neq p', \hspace{20pt} q\neq q'.
\end{split}
\end{equation}
The introduction of ancillary systems overcomes this peculiarity of 
graded algebras. Indeed, by paring with an ancillary cell $\M{n}{n}$
we have:
\begin{align*}
&\M{p}{q}\gt\M{n}{n}\simeq \M{n(p+q)}{n(p+q)},\\
&\M{p'}{q'}\gt\M{n}{n}\simeq \M{n(p'+q')}{n(p'+q')}.
\end{align*}
Hence they are both isomorphic to $\M{2n^2}{2n^2}$ due to
\Eq\eqref{eq:dimalgebra}. Therefore only one ancillary system for each 
cell, isomorphic to the cell itself, is sufficient to enable the same 
argument as for QCAs. 
\end{proof}

\begin{remark}
\label{rem:n=1}
In the special case $n=1$, the only non trivial subalgebra of
$\M{1}{1}$ is $\M{1}{1}$ itself. Thus, every support algebra is
immediately isomorphic to each other without appending an auxiliary
ancillary system $\M{1}{1}$. Therefore, in the case of a single local
Fermionic mode per site, we may expect that index one FCA over
$\M{1}{1}$ should satisfy the $\mathcal{M}$-implementability condition
without the addition of ancilla degrees of freedom. This last
statement shall be precisely provided in Section \ref{sec:Class}
through directly classifying index one FCA over $\M{1}{1}$.
\end{remark}
Thanks to this statement it holds that:
\begin{corollary}
  Every FCA $\tT$ with $\ind{\tT}=1$ is $\tM$-implementable over the
  enlarged system with ancillae, as in \Eq\eqref{eq:Margolus}, i.e.
\begin{align}\label{eq:Manc}
\tT\gt\tI=\tM,
\end{align}
for some MS $\tM$,  with only one ancilla per site.
\end{corollary}
\begin{proof}
  From the previous Proposition we can deduce that each $\ind{\tT}=1$
  automata $\tT$ is stably-equivalent to the identity with only one
  ancilla per site. The proof then trivially follows considering the
  form of the stable-equivalence in \Eq\eqref{eq:Mequiv}.
\end{proof}
Based on this, we can show how to get rid of the notion of stable
equivalence for index one FCA. To this aim, we can exploit
\Eq\eqref{eq:stabeq} with a single ancilla per cell to implement a new
circuit $\tF'$ explicitly which uses cells of the system to play the
role of the ancillae. Firstly, we apply the gates of the circuit $\tF$
over a portion $R$ of the lattice. In the meantime, we choose a set of
cells not involved in the transformation of $R$ to play the role of
the ancillae. Since $\tF$ acts trivially on the ancillae, once $R$ is
fully updated the cells employed in the previous step are now in the
same state as before the updating of $R$. This procedure is iterated
until the whole lattice is updated. This result is stated in the
following proposition.

\begin{proposition}
	\label{prop:anc}
	An index one FCA is $\tF$-implementable according to Definition \ref{def:implem}.
\end{proposition}

\begin{proof}
  Consider the $\tM$-implementability on the enlarged system of an
  index one FCA as in \Eq\eqref{eq:Manc}.  An explicit procedure for
  a single ancillary system may be graphically provided as follows.

\begin{itemize}
\item {\em Step 1}. Consider a sub lattice of 12 cells, labelled $\{0,...,11\}$. We act with the first layer of gates over the cells in the range $2-7$. All the other available cells play the role of the ancillae.\\
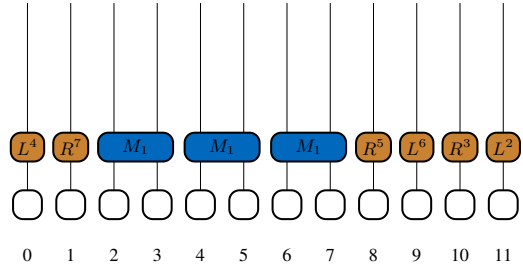
\begin{figure}[h!]
\resizebox{0.8\hsize}{!}{\input{figures/Ancillaremcircuit0.tikz}}
\caption{The first layer of gates ($M_1$) is colored in blue. The
  ancilla's transformations associated with the left (L) and right (R)
  leg of the $M_1$ transformations are drawn in orange. The apex
  denotes the number of the cell that is coupled with the ancilla by
  the transformation.}
\end{figure}
\item {\em Step 2}. We act with the second layer of gates over the
  cells 3-6. Since this step corresponds to a complete evolution of
  the 3-6 cells, the corresponding ancilla's transformations are
  removed from the circuit.
\begin{figure}[h!]
\resizebox{0.8\hsize}{!}{\input{figures/Ancillaremcircuit1.tikz}}
\caption{The second layer of transformations ($M_2$) provided in the
  second step is colored in light blue. After the action of such
  transformations, only ancilla's transformation relative to the cell
  7 and 2 are left. All the other cells are now free for the next
  steps.}
\end{figure}
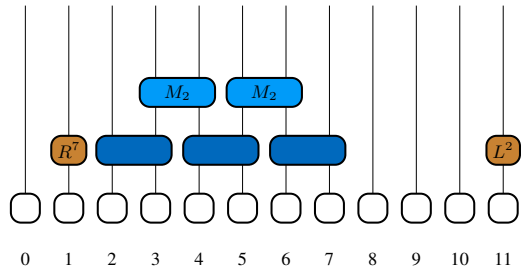
\item {\em Step 3}. We apply subsequently a transformation $M_1$
  over 8-9 and a transformation $M_2$ over 7-8. To this aim, we add
  and remove an ancilla's transformation relative to the cell 8. The
  cell on which this transformation acts is immaterial. The ancilla's
  transformation of the 7 cell is removed while a new one relative to
  the cell 9 is added on already fully updated cell 3.
\item {\em Step 4}. Similarly to the Step 3, a transformation $M_1$
  acts over 0-1, and a transformation $M_2$ over 1-2. Again, the
  ancilla's transformation relative to the cell 1 is added and removed
  in the same step, while an ancilla transformation is added for the
  cell $0$ on the already fully updated cell 4.
\begin{figure}[h!]
\resizebox{0.8\hsize}{!}{\input{figures/Ancillaremcircuit2.tikz}}
\resizebox{0.8\hsize}{!}{\input{figures/Ancillaremcircuit3.tikz}}
\caption{Step 3 and Step 4. Notice that the ancilla's transformations
  act over the fully updated cells 3-4. In this way all the cells not
  yet updated are now free.}
\end{figure}
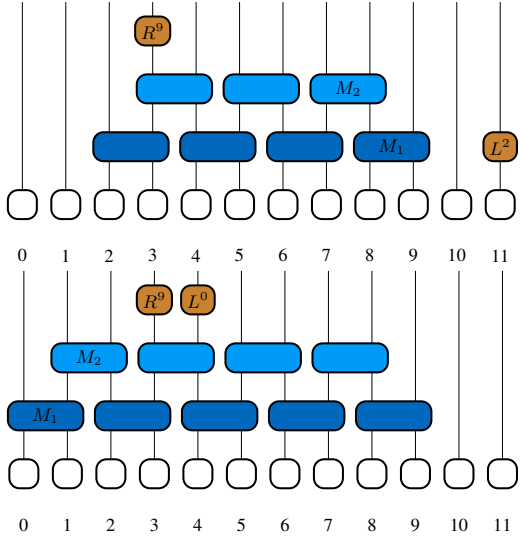
\item {\em Step 5}. All the missing transformations are applied and the full FDFC is implemented.  
\begin{figure}[h!]
\resizebox{0.8\hsize}{!}{\input{figures/Ancillaremcircuit4.tikz}}
\caption{After Step 5 the FDFC is complete. The picture corresponds to
  the circuit $\tF$ in \Eq\eqref{eq:ancillarem}. The action of the
  ancilla is obvious and is not specified.}
\end{figure}
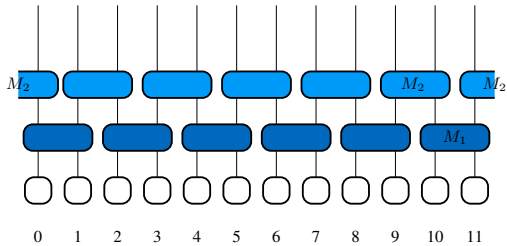
\end{itemize}
\end{proof}

\begin{remark}
  Based on the proof of Proposition \ref{prop:anc}, the quantum
  circuit $\tF$ obtained in the ancilla removal procedure is not a
  nearest neighbour FDFC, and thus it exhibits a different neighbourhood
  scheme than the FDFC $\tilde{\tF}$ in Definition
  \ref{def:stableeq}. In addition, such a procedure is not unique, and
  hence it may yield several FDFC with different
  neighbourhood schemes or depths. 
\end{remark}

Proposition \ref{prop:anc} has a trivial but remarkable implication
at the level of equivalence of FCA, which extend their ungraded
counterpart to the graded case. It tells us that FCA sharing the same
index belong to the same equivalence class defined by Definition
\ref{def:equiv}.
\begin{corollary}\label{cor:equivindex}
  Let $\tK: \Aqloc \rightarrow \Aqloc$, and
  $\tJ: \Aqloc \rightarrow \Aqloc$ be two FCA with
  $\ind{\tK}=\ind{\tJ}$. Then $\tK$ and $\tJ$ are $\tF$-equivalent
  according to Definition \ref{def:equiv}, i.e., there exists an
  FDFC $\tF$ such that $\tK = \tF \circ \tJ$. Moreover, upon suitable 
  regrouping of cells one can recast $\tF$ in an MS.
\end{corollary}
  \begin{proof}
    The Proof consist in a trivial application of
    \Eq\eqref{eq:Manc}. Indeed, consider two FCA $\tK,\tJ$ with
    $\ind{\tK}=\ind{\tJ}$. From the index properties it trivially
    follows that $\ind{\tK\circ\tJ^{-1}}=1$. Then we have:
\begin{align*}
\tK\circ\tJ^{-1}\gt\tI=\tM.
\end{align*}
We can apply the ancilla removal to get:
\begin{align*}
\tK\circ\tJ^{-1}=\tF.
\end{align*}
Composing on the right with $\tJ$ we finally get $\tK=\tF\circ\tJ$.
Finally, one can invoke Lemma~\ref{lem:Mequiv} to obtain, upon suitable regrouping of the cells, $\tK=\tM\circ\tJ$ for some MS $\tM$.
\end{proof}

\section{Classification of One dimensional Nearest Neighbours FCA}
\label{sec:Class}
In this section, we classify all the possible evolutions of FCA on a
one-dimensional lattice $\Lambda$ with a single Fermionic mode per
site, namely the support algebra per site is isomorphic to
$\M{1}{1}$.

We recall that each local Fermionic mode located at $x$ is
generated by $(X_x,Y_x)$. Based on Theorem $\ref{thm:locglo}$, the
form of the automorphism implementing the FCA is pinned down only by
the graded commutation relations on the overlapping regions of the
neighbourhoods. In particular, for a nearest neighbour FCA 
$\tT$, i.e. with $\mathcal{N}=\{-1,0,1\}$, it holds
that

\begin{equation}
	\label{eq:incinc}
        \tT(\mathsf{A}_x)\subseteq\mathsf{A}_{x-1}\boxtimes\mathsf{A}_x\boxtimes\mathsf{A}_{x+1}.
\end{equation}
Considering for example the region
$\mathsf{A}^0\boxtimes\mathsf{A}^1\boxtimes\mathsf{A}^2\boxtimes\mathsf{A}^3\boxtimes\mathsf{A}^4$,
it is sufficient to check the two conditions
\begin{align}
\{[\tT(\mathsf{A}_1),\tT(\mathsf{A}_3)]\}=0,\\
\gc{\tT(\mathsf{A}_1)}{\tT(\mathsf{A}_2)}=0.
\label{eq:condcomm}
\end{align}
In particular, we can define the support algebras:
\begin{equation}
\label{eq:support}
	\begin{aligned}
		\EL &\coloneqq
                \mathbf{S}(\tT(\mathsf{A}_x),\mathsf{A}_{x-1}),\\
                \ER &\coloneqq \mathbf{S} ( \tT ( \mathsf{A}_x ), \mathsf{A}_{x+1}),\\
		\EC &\coloneqq \mathbf{S}(\tT(\mathsf{A}_x),\mathsf{A}_x).
	\end{aligned}
\end{equation}
In the translationally invariant case, those support algebras do not depend
on $x$. Moreover, using \Eq\eqref{eq:support},  \Eq\eqref{eq:incinc} becomes
\begin{equation}\label{eq:incl}
\tT(\m A_x)\subseteq \EL\gt\EC\gt\ER.
\end{equation}
Lemma $\ref{thm:suppcomm}$ allows us then to derive the following necessary condition from \Eq\eqref{eq:condcomm}
\begin{equation}
\label{eq:grc}
	\gc{\EL}{\ER}=0.
\end{equation}

Another preliminary result for the scope of the classification of 
FCA is the following Lemma and along with its corollary.

\begin{lemma}\label{lem:diffpar}
 Operators in the same cell with different parity that are not multiples of $I$ cannot graded-commute
\end{lemma}
\begin{proof}
  The thesis immediately follows considering that odd operators are of the form
  $aX+bY$ and even operators are of the form $cI+dZ$, and that
\begin{equation*}
\gc{X}{Z}=[X,Z]=-iY\neq 0.
\end{equation*}
\end{proof}

The previous Lemma has two relevant corollaries.

\begin{corollary} \label{cor:fullmat} 
The graded center of the full matrix graded 
algebra $\mathcal{M}$ is trivial.
\end{corollary}
\begin{corollary}\label{cor:par} One of the following cases holds
\begin{enumerate}
\item
$\EL$ (or $\ER$) is trivial;
\item
$\EL=\ER=\{cI+dZ\mid c,d\in\mathbb C\}$;
\item
with a suitable choice of basis, $\EL=\{aI+bX\mid a,b\in\mathbb C,\ ab=0\}$ and 
$\ER=\{cI+dY\mid c,d\in\mathbb C,\ cd=0\}$.
\end{enumerate}
\end{corollary}
\begin{proof}
This straightforwardly follows from the above Lemma and the first condition in 
\Eq\eqref{eq:grc}. One needs to consider that the only subalgebra of $\M{1}{1}$ 
with fixed parity---thus graded commuting with itself---is 
$\{cI+dZ\mid c,d\in\mathbb C\}$. On the other hand if $\ER$ and $\EL$ contain odd selfadjoint operators, say $O_L=aX+bY$ and $O_R=cX+dY$ with $a,b,c,d\in\mathbb R$, 
then
\begin{align*}
\gc{O_L}{O_R}=2(\mathbf a\cdot\mathbf c)I,
\end{align*}
where $\mathbf a=(a,b)$ and $\mathbf c=(c,d)$, thus they graded commute iff 
$\mathbf a= \|\mathbf a\|R\mathbf e_1$ and $\mathbf c= \|\mathbf c\|R\mathbf e_2$ with 
$(\mathbf e_i)_j=\delta_{ij}$ and $R$ a rotation matrix. This implies that 
$O_L=\|\mathbf a\|X'$ and $O_R=\|\mathbf c\|Y'$ with 
$W'\coloneqq U(\theta, n)^\dag WU(\theta,n)$.
\end{proof}

The classification of all one-dimensional homogeneous FCA $\tT$ over
cells isomorphic to $\M{1}{1}$ articulates in three steps: i) compute
all admissible values of the index; ii) show that shifts or
Majorana-shifts exhaust all the values of index different from 1; iii) 
characterise all index one FCAs.

We start computing all possible values of the index. To this end we
consider \Eq\eqref{eq:indform} where we take
\begin{equation*}
\begin{aligned}
	&d=2, && m = 1,2, \\
	&p+q=2^k, && k=\{0,1,2\},
\end{aligned}
\end{equation*}
which gives the following values of the index
\begin{equation*}
\ind{\tT}=1,2^{\pm\frac{1}{2}},2^{\pm 1}.
\end{equation*}

Examples of quantum cellular automata $\tG$ having
$\ind{\tG}=1,2^{\pm 1}$ are easily found (even in the
ungraded). Indeed, each FDFC has $\ind{\tF}=1$, while any shift
$\tau_x$ of a
$\mathsf{A}_{x}=\operatorname{Mat}_{2\times 2}(\mathbb{C})$ matrix
algebra has $\ind{\tau_x}=2^{\pm 1}$. Finally, irrational indices like
$\ind{\tT}=2^{\pm\frac{1}{2}}$, which are peculiar of the Fermionic
graded case, can be obtained via Majorana
shifts~\eqref{def:Majoranashift} as proved in the following lemma:

\begin{lemma}
The Majorana shift $\sigma_{\pm}$ in \Eq \eqref{eq:Majoshift} satisfies
\begin{align}\label{eq:ind-Majoshift}
&\operatorname{ind}(\sigma_{\pm})=2^{\pm\frac{1}{2}},\\
&\operatorname{ind}(\sigma_{\pm}^{-1})=2^{\mp\frac{1}{2}}.
\end{align}
\end{lemma}
\begin{proof}
  Applying $\sigma_\pm$ twice, as defined in \Eq\eqref{eq:Mshift}, one can easily find 
  that $(\sigma_{\pm})^2=\tau_{\pm}$. Hence, exploiting the
  multiplicative property of the index in \Eq\eqref{eq:indmult} we  have
\begin{equation}
\ind{(\sigma_{\pm})^2}=\ind{(\sigma_{\pm})}^2=\ind{\tau_\pm},
\end{equation}
and thus \Eq\eqref{eq:ind-Majoshift} trivially follow. 
\end{proof}

By the multiplication property in Lemma \ref{lem:prop} every two
FCA $\tT,\tS$ with $\ind{\tT}=\ind{\tS}$ are related through a FDFC,
namely an index one FCA. On the other hand shifts $\tau_{\pm 1}$ and
Majorana shifts $\sigma_{\pm}$ and $\sigma_\mp^{-1}$ are good representatives for index
equal to $2^{\pm 1}$ and $2^{\pm 1/2}$, respectively. Therefore we conclude
that
\begin{align}
&  \ind{\tT}=2^{\pm 1}\Rightarrow \tT=\tau_{\pm 1}\circ\mathcal{\tS},\\
&  \ind{\tT}=2^{\pm \frac{1}{2}}\Rightarrow \tT=\sigma_{\pm}\circ \mathcal{S},
\end{align}
for a suitable FCA $\tS$ with $\ind{\tS}=1$.  As a result we only need
to provide a characterization of index one FCA.

\subsection{Unit Index FCA}
We begin proving that for index one FCA over $\M{1}{1}$ the flow of
information to the left and to the right is balanced:
\begin{lemma}
For FCA with cells isomorphic to $\M{1}{1}$ and $\ind{\tT}=1$ one has
$\EL\simeq\ER$.

\end{lemma}
\begin{proof}
  We proceed exhausting all cases. Without loss of generality, we will 
  enumerate all possibile forms of $\EL$: the same arguments clearly hold
  exchanging  $\EL$ and $\ER$.
\begin{enumerate}
\item\label{it:full} $\EL\simeq\M{1}{1}$. The graded commutation relation in
  \Eq\eqref{eq:grc} and Corollary \ref{cor:fullmat} immediately
  imply $\ER\simeq I$. Thus, the second relation in
  \Eq\eqref{eq:condcomm} reduces to $\gc{\EL}{\EC}=0$, which implies
  $\EC\simeq \mathbb CI$. Therefore, the automaton is isomorphic to a left
  shift which has not unit index. 
\item\label{it:odd} $\EL$ contains an odd operator, that we will call $X$. In
  this case, the condition in \Eq\eqref{eq:grc} and Theorem
  \ref{lem:diffpar} imply that either $\ER$ is generated by an odd
  operator that anticommutes with $X$---namely $Y$---or it is the
  trivial algebra $\ER\simeq I$. In the former case we have indeed
  $\ER\simeq \EL$, while in the latter we can apply the second
  condition in \Eq\eqref{eq:condcomm} along with Corollary~\ref{cor:par}
  to prove that $\EC$ is either
  trivial or generated by $Y$. On the one hand, if $\EC$ is generated
  by $Y$, the evolution is isomorphic to the {Majorana shift}
  defined in \eqref{eq:Mshift}, and thus it has index different from
  one. On the other hand, if $\EC$ is trivial, we would have
  $\ER\simeq \EC \simeq \mathbb CI$ and $\EL$ generated by $X$. However,
  \Eq\eqref{eq:incl} implies that
\begin{equation*}
	\dim{\EL}\dim{\ER}\dim{\EC}\ge\dim{\m A_x},
\end{equation*}
and this last case is then ruled out by dimensional arguments.
\item\label{it:even} $\EL$ is non-trivial and even. Then, $\ER$ is
  itself non-trivial and even, or it is trivial. By similar arguments
  as those at point $2$, the former case is the only admissible
  solution for an index one FCA, and thus 
  $\EL\simeq \ER\simeq\{cI+dZ\mid c,d\in\mathbb C\}$.
\item $\EL\simeq \mathbb CI$. All the admissible cases are listed as follows:
  $\ER\simeq \M{1}{1}$, $\ER$ is $\{cI+dZ\mid c,d\in\mathbb C\}$, $\ER$ 
  contains an odd operator---say $X$---or $\ER$ is trivial as well. The first three 
  cases are ruled out by exchanging the roles of $\EL$ and $\ER$ in items~\ref{it:full},
  \ref{it:odd} and~\ref{it:even}. Therefore, we are only left with 
  $\ER\simeq \EL\simeq \mathbb CI$.
\end{enumerate}
\end{proof}
The classification proceeds as follows. Based on the previous Lemma,
three different cases must be studied:
\begin{enumerate}
\item \label{it:triv}$\EL\simeq \ER\simeq \mathbb CI$;
\item \label{it:ev}$\EL\simeq\ER\simeq\{cI+dZ\mid c,d\in\mathbb C\}$;
\item \label{it:o}
$\EL=\{aI+b O_L\mid a,b\in\mathbb C,\ ab=0\}\simeq\ER=\{aI+b O_R\mid a,b\in\mathbb C,\ ab=0\}$, with $\gc{O_L}{O_R}=0$.
\end{enumerate}
Let us start with the trivial case~\ref{it:triv}. Since the evolution of an FCA
is an automorphism, the only possible choice for the support algebra $\EC$
is given by $\EC\simeq\M{1}{1}$, i.e. the full matrix graded
algebra. Therefore the evolution consists of a trivial embedding of
the full matrix graded algebra in the central cell, corresponding to a
local unitary transformation (see \Eq\eqref{eq:unit}).

Let us now focus on case~\ref{it:ev}. This is analogous to the qubit case in
\cite{schumacher2004reversible}, with the addition of imposing the
parity superselection rule stated in Theorem~\ref{thm:Parity}. The
results are summarized in the following proposition proved in
Appendix~\ref{sec:proof}.

\begin{proposition}\label{prop:Werner}
  Let $\tT_0: \Aqloc\rightarrow\Aqloc$ be the local rule of a FCA
  having $\EL$ and $\ER$ both generated by $Z$. 
Denote by
\begin{align}
   G_{\phi_k} = (C_{\phi_k}\gt I_3)(I_1\gt C_{\phi_k}),
\end{align} 
the unitary operator which implements the controlled-phase gate
\begin{align}
\label{eq:dephasing}
  C_{\phi_k} =\begin{pmatrix} 1 &0&0&0\\0&1&0&0\\0&0&1&0\\0&0&0&e^{i\phi_k}\end{pmatrix}, \qquad \phi_k \in [0,2\pi),
\end{align}
diagonal in the $Z$-basis. Then the action of $\tT_0$ is given by
\begin{equation}
	\label{eq:ev}
		\begin{aligned}
	\tT_0(\xi)&= G^{\dag}_{\phi} (I \gt \mathscr{U}^{(\theta,n)}(\xi) \gt I)G_{\phi},
		\end{aligned}
\end{equation}
for some local unitary $\mathscr{U}^{(\theta,n)}$.
\end{proposition}
Finally, case~\ref{it:o} is studied in the following theorem and gives rise
to a FCA, denoted \emph{forking automaton} , with no ungraded
counterpart. The proof of the theorem is in Appendix
\ref{sec:proof}.

\begin{theorem}(Forking automaton)\label{thm:main} Let
  $\tT_0: \Aqloc\rightarrow\Aqloc$ be the local rule of a FCA having
  $\EL$ and $\ER$ generated by two odd anticommuting operators. Then
  there exists a pair of odd anticommuting operators $(\xi,\eta)$ such
  that the action of $\tT_0$ over $(\xi,\eta)$ can be expressed as:
\begin{equation}\label{eq:forking}
\begin{aligned}
\tT_0(\eta)&= X\gt I \gt I,\\
\tT_0(\xi)&=I\gt I \gt Y.
\end{aligned}
\end{equation}
The FCA having the above local rule can also be written as 
$\tT=\tilde\tT\circ\tU$ where $\tU(O_x)=\tU_x^{(\theta,n)}(O_x)$ is such that 
$\tU^{(\theta,n)}(\xi)=X$ and $\tU^{(\theta,n)}(\eta)=Y$,
while
\begin{equation}
	\tilde\tT_0: (X_{x},Y_{x}) \mapsto (Y_{x-1},X_{x+1}).
\end{equation}
We call $\tT$ a \emph{forking} FCA.
\end{theorem}

As for the Majorana shifts~\ref{def:Majoranashift}, also the forking
FCA splits odd Fermionic degrees of freedom at site $x$
and spreads them on neighbouring sites in a way that qubits automata cannot
achieve. While Majorana shifts (as well as shifts) cannot be locally
implemented (they do not have unit index), the forking
FCAs~\Eq\eqref{eq:forking} are $\tM$-implementable according to
Definition \ref{def:implemloc}. In the following corollary, proved in Appendix \ref{sec:Cor},
we provide an explicit MS for the forking FCA:

\begin{corollary}\label{cor:Fork}
  The forking automaton $\tT_0$~\Eq\eqref{eq:forking} can be
  implemented through the following Margolus Partitioned Scheme:
\begin{align}
  &\tT_0=\tM_2\circ\tM_1\label{eq:marpasc}\\
  &\tM_2(\cdot)=\prod_x(U\gt U) M_{2x}(\cdot)\prod_y M_{2y}^\dag (U\gt U)^\dag	\\ &\tM_1(\cdot)=\prod_x M_{2x+1}(\cdot)\prod_y M^\dag_{2y+1},
\end{align}
where
\begin{align*}
M_{x}=(X\gt Z)e^{-\frac{\pi}{4}Y_{x}\gt X_{x+1}},
\end{align*}
and $U=U(\theta,n)$ as in \Eq\eqref{eq:unit}.
\end{corollary}

The characterization of index one FCA over $\M{1}{1}$ is now
complete. It shows that all the index one FCA are given either by
the QCA over qubits obtained in \cite{schumacher2004reversible} (see
Proposition \ref{prop:Werner}) or by a forking FCA (see
Theorem \ref{thm:main}). Furthermore, it was proved that forking
FCAs are $\tM$-implementable (see Corollary \ref{cor:Fork}). To
conclude this analysis we state the following results about
one-dimensional FCA over $\M{1}{1}$, which provides a stronger version
of the statements contained in Section \ref{sec:stabequiv} (see also
Remark \ref{rem:n=1}).
\begin{theorem}
  Every index one FCA over $\M{1}{1}$ is $\tM$-implementable
  according to Definition \ref{def:implemloc}. Moreover, two FCA
  over $\M{1}{1}$ with the same index are $\tF$-equivalent according
  to Definition \ref{def:equiv}.
\end{theorem}
\begin{proof}
  The first statement is immediately obtained by the classification of
  index one FCA over $\M{1}{1}$ performed before. On the other hand,
  the second statement directly follows from \Eq\eqref{eq:indmult} in
  Lemma \ref{lem:prop}. Indeed, two FCA with the same index are
  related by an index one FCA, which has been proven to be
  $\tF$-implementable. Therefore, they are $\tF$-equivalent through
  a scheme as in \Eq\eqref{eq:marpasc}.
\end{proof}

\section{Outlook}
One dimensional quantum cellular automata, representing the dynamics of chains
of locally interacting quantum systems (e.g.~ qubits in the simplest case),
have been extensively studied under various aspects. In particular, an interesting
facet of the theory of QCAs is based on their local algebraic 
tensor product structure. Operators localised on disjoint regions of
the lattice commute with each other and the evolution spreads them over finite 
regions, preserving commutation relations. The analogous scheme for local Fermionic
modes, governed by CAR algebras, is more tricky, since the
simple tensor product structure has to be replaced by graded tensor product. 
This issue is typically bypassed referring to concrete representations of 
CAR algebras on qubits (i.e. the Jordan-Wigner representation), which 
makes it possible to use the standard concepts of quantum information theory, 
subsequently imposing parity superselection rule. At first sight this approach seems 
to confine Fermionic cellular automata to a subset of the admissible
local quantum dynamics.

On the contrary, the classification of local dynamics for Fermionic
modes~\cite{PhysRevB.99.085115} brought out core-level
discrepancies with respect to its standard quantum counterpart. In
both cases, cellular automata are divided in equivalence classes labelled by
an index, representing the unique topological invariant for one
dimensional lattices. The values of the index, which encode the ratio of 
information flow of a given cellular automaton in the two directions on the 
lattice (right/left), span all the rational numbers in the quantum case, 
while for Fermionic systems it can take irrational values. 
This is reminiscent of the local graded-commuting structure of Fermionic algebras, 
which allows for moving (Majorana shifts) odd Fermionic degrees of freedom that 
lack a qubits analogue.

The above classification is given modulo finite depth Fermionic circuits
for QCAs~\cite{Werner-index}, and modulo finite depth Fermionic circuits
with the addition of ancillae for FCAs~\cite{PhysRevB.99.085115}. Via
ancilla removal~\cite{Freedman-Haah-Hastings-group-structure}, we
proved that also in the Fermionic case the equivalence relation boils
down to the quantum one, namely two FCAs have the same index if and
only if they are connected by a Fermionic circuit of finite depth.

As an outcome of Fermionic index properties, every FCA is
equivalent to a unit index automaton, which has balanced right/left
flow of information, composed with a Fermionic shift, introducing an
arbitrary degree of imbalance. Accordingly, an exhaustive
parametrization of Fermionic automata follows from characterizing all
unit index representatives. This is achieved in this paper, showing a
discrepancy between quantum and Fermionic automata already within
locally implementable instances. In other words, there are finite depth 
Fermionic circuits that cannot be implemented as finite depth qubit circuits.
In Ref.~\cite{schumacher2004reversible} is was shown that every QCA with
circuital realization is made of single system and controlled-phase
gates. On the other hand, Fermionic modes allow for a localizable
automaton that does not admit such a realization.

\section{Aknowledgments}
PP and PM acknowledge financial support from European Union—Next
Generation EU through the MUR project Progetti di Ricerca d’Interesse
Nazionale (PRIN) QCAPP No. 2022LCEA9Y. AB acknowledges financial
support from European Union—Next Generation EU through the MUR project
Progetti di Ricerca d’Interesse Nazionale (PRIN) DISTRUCT
No. P2022T2JZ9. LT acknowledges financial support from European Union—
Next Generation EU through the National Research Centre for HPC, Big
Data and Quantum Computing, PNRR MUR Project
CN0000013-ICSC. A. T. acknowledges the financial support of Elvia and
Federico Faggin Foundation (Silicon Valley Community Foundation
Project ID$\#2020-214365$).

\bibliographystyle{unsrt}
\bibliography{FQCA}

\appendix

\section{Proof of Proposition~\ref{prop:Werner}}\label{sec:proof}
We want now to prove Proposition \ref{prop:Werner}. For convenience
of exposition, the proof will be splitted in three parts. The first
part, gathered in Lemma \ref{lem:resemb}, shows how the
graded-commutation relation reduces to a decomposition of the FCA
similar to the ungraded QCA obtained in
\cite{schumacher2004reversible}. Then we prove a Lemma which is the
Fermionic analog of Lemma 9 in \cite{schumacher2004reversible}, and
finally we show that this formal resemblance actually leads to the
same evolutions of the QCA in the ungraded case, constrained by the
Parity Superselection Rule.
\begin{lemma}\label{lem:resemb}
  Consider a FCA $\tT$ whose support algebras $\EL$ and $\ER$ are
  both generated by an even operator, say $Z$. Given two arbitrary odd
  operators $\xi, \eta$, such that $\acomm{\xi}{\eta}=0$, then the
  form of the evolution of $\tT$ is given by:
\begin{align}
&\tT(\xi)=\sum_{i,j=0}^1 \ketbra{i}{i}\gt U_{i,j}(\xi)U_{i,j}^\dag\gt \ketbra{j}{j},\label{eq:Tform} 
\end{align}
where $U_{i,j}\in SU^{n}(2)$, $n=0,1$, satisfy the following constraint $\forall k,l=0,1$
\begin{align}
  &\sum_{i,j=0}^1\gc{ U_{k,i}(\xi)U_{k,i}^\dag\gt \ketbra{j}{j}}{\ketbra{i}{i}\gt U_{j,l}(\eta)U_{j,l}^\dag}=0.\label{eq:Tcomm}
\end{align}
\end{lemma}
\begin{proof}
  Since $\EL,\ER$ are generated by $Z$, we can decompose the action on
  an odd operator $\xi$ as:
\begin{equation}\label{eq:decz}
\tT(\xi)=\sum_{i,j=0}^1 \ketbra{i}{i}\gt \mathsf{A}_{i,j}(\xi)\gt \ketbra{j}{j},
\end{equation}
where $\ket{i},\ket{j}$ are eigenstates of $Z$.
To impose the preservation of the CAR, it should hold that  $\{\tT(\xi),\tT(\eta)\}=0$, and thus:
\begin{equation}
\sum_{i,j=0,1} \ketbra{i}{i} \gt \{\mathcal{A}_{i,j}(\xi),\mathcal{A}_{i,j}(\eta)\}\gt \ketbra{j}{j}=0 ,
\end{equation}
which yields to:
\begin{equation}
\{\mathcal{A}_{i,j}(\xi),\mathcal{A}_{i,j}(\eta)\}=0 \hspace{10pt} \forall i,j\in\{0,1\}.
\end{equation}
This means that also $\mathcal{A}_{i,j}(\xi)$ and
$\mathcal{A}_{i,j}(\eta)$ are representation of the CAR. Since every
two representations of the CAR are related through a unitary
transformation we can write $\mathcal{A}_{i,j}(\xi),\mathcal{A}_{i,j}(\eta)$ as a unitary
rotation of $\xi,\eta$ with defined parity, i.e.
\begin{equation}
  \mathcal{A}_{i,j}(\xi)=U_{i,j}\xi U_{i,j}^\dag \hspace{5pt} U_{i,j}\in SU^{n}(2) \hspace{5pt} \forall i,j \in\{0,1\},
\end{equation}
for a fixed $n=0,1$. Clearly, the same holds for $\eta$. We have then:
\begin{equation}
\tT(\xi)=\sum_{i,j=0}^1 \ketbra{i}{i}\gt U_{i,j}(\xi)U_{i,j}^\dag\gt \ketbra{j}{j}. \label{eq:abev}
\end{equation}
On this evolution, we still have to impose the graded commutation
relation $\gc{\tT(\xi)\gt I}{I\gt\tT(\eta)}$. This leads to:
\begin{align*}
&\gc{\tT(\xi)\gt I}{I\gt\tT(\eta)}=\\
&=\sum_{i,j=0}^1\gc{ U_{k,i}(\xi)U_{k,i}^\dag\gt \ketbra{j}{j}}{\ketbra{i}{i}\gt U_{j,l}(\eta)U_{j,l}^\dag}=0,
\end{align*}
$\forall k,l=0,1$ and this concludes the proof.
\end{proof}

\begin{lemma}
  Let $i,j=0,1$ and $\{U_i\}_{i=1}^2,\{V_{j}\}_{j=1}^2$ be two sets of
  unitary $2\times 2$ matrices satisfying:
\begin{equation*}
\sum_{i,j=0}^1\{ U_{i}(\xi)U_{i}^\dag\gt \ketbra{j}{j},\ketbra{i}{i}\gt V_{j}(\eta)V_{j}^\dag\}=0.
\end{equation*}
Then $U_iU_j^\dag,V_iV_j^\dag$ are simultaneously diagonalizable in
the $\ketbra{i}{i}$ eigenstate basis of $Z$.
\end{lemma}
\begin{proof}
Let $E_i=U_i\xi U_i^\dag$ and $F_j=V_j\eta V_j$. We can cast the CAR as:
\begin{equation*}
\sum_{i,j} E_i\ketbra{i}{i}\gt \ketbra{j}{j}F_j=\sum_{i,j}\ketbra{i}{i} E_i\gt F_j\ketbra{j}{j}.
\end{equation*}
Notice that we exploited
$\sum_{i,j}(\ketbra{i}{i} \gt
F_j)(E_i\gt\ketbra{j}{j})=-\sum_{i,j}\ketbra{i}{i} E_i\gt
F_j\ketbra{j}{j}$, i.e. the anticommutation of $E_i$ and $F_j$ on
different sites.  To simplify the sum, we can multiply on the left
with $\ketbra{l}{l}\gt\ketbra{m}{m}$ and on the right with
$\ketbra{a}{a}\gt\ketbra{b}{b}$. In this way, we get :
\begin{align*}
&\ketbra{l}{l} E_a \ketbra{a}{a}\gt \ketbra{m}{m}F_m\ketbra{b}{b}=\\
&=\ketbra{l}{l}E_l\ketbra{a}{a}\gt\ketbra{m}{m}F_b\ketbra{b}{b}.
\end{align*}
Again, we can exploit the invariance of the classification under local
unitary rotations, and then choose $F_b$ in such a way that
$\bra{m}F_b\ket{b}=0$. In particular, we can choose
$\eta=V^\dag_b \ket{m'}\bra{b'}V_b$, with $m'\neq m$. Thereby, the
left hand side of the previous equation vanishes and we get:
\begin{equation}
  \ketbra{l}{l} E_a \ketbra{a}{a}\gt \ketbra{m}{m}V_mV_b^\dag\ketbra{m'}{b'}V_bV_m^\dag\ketbra{b}{b}=0.
\end{equation}
Finally, we notice that we can always choose $\xi$ in such a way that
$\ketbra{l}{l} E_a \ketbra{a}{a}\neq 0$ and $b'$ in such a way that
$\bra{b'}V_bV_m^\dag\ket{b}\neq 0$. Thus, this leads to:
\begin{equation*}
\bra{m}V_mV_b^\dag\ket{m'}=0,
\end{equation*}
which proves the lemma for $V_mV_b^\dag$ $\forall m,b$. Repeating the
same procedure with the roles of $V$ and $U$ exchanged, we get to the
same conclusion for $U_iU_j^\dag$.
\end{proof}
Thanks to this last lemma, we can now prove the main Proposition \ref{prop:Werner}.
\begin{proof}
  We start from the commutation relation in \Eq\eqref{eq:Tform} and
  \eqref{eq:Tcomm} and prepend to $\tT$ a local transformation given
  by the conjugation with $U_{\tilde{i},\tilde{j}}^\dag$ for a fixed
  value of $\tilde{i},\tilde{j}$. In this way we get to:
\begin{align*}
  \tT(\xi)=\sum_{i,j=0}^1 \ketbra{i}{i}\gt U^\dag_{\tilde{i},\tilde{j}}U_{i,j}(\xi)U_{i,j}^\dag U_{\tilde{i},\tilde{j}}\gt \ketbra{j}{j},
\end{align*}
After redefining
$\tilde{U}_{i,j}= U^\dag_{\tilde{i},\tilde{j}}U_{i,j}$, using the
previous lemma, all the $\tilde{U}_{i,j}$ are simultaneously
diagonalizable. Moreover, we can compose with a second local rotation
in such a way that $\tilde{U}_{0,0}=I$. This allows us to write:
\begin{align*}
U_{i,j}=e^{i\phi_{i,j}}\ketbra{0}{0}+e^{i\phi_{i,j}}\ketbra{1}{1} && U_{0,0}=I,
\end{align*}
where we dropped the symbol  $\hspace{3pt}\tilde{ }\hspace{3pt}$ to ease the notation. Inserting
$\xi=\ketbra{\alpha}{\alpha\oplus 1}$,
$\eta=\ketbra{\beta\oplus 1}{\beta}$ in \Eq\eqref{eq:Tcomm}, where
$\oplus$ is the sum modulo 2 and $\alpha,\beta=0,1$, we get to:
\begin{equation*}
  (-1)^\alpha(\phi_{k,\beta}-\phi_{k,\beta\oplus 1})=(-1)^\beta(\phi_{\alpha,l}-\phi_{\alpha\oplus 1,l})\hspace{3pt} \operatorname{mod} 2\pi
\end{equation*}
$\forall \alpha,\beta,k,l\in\{0,1\}$. In particular, we know that
$\phi_{0,0}=0$ as a trivial consequence of the fact that
$U_{0,0}=I$. Evaluating for $k=l=0$ and for $k=l\oplus 1$ we get:
\begin{align*}
\phi_{1,1}=2\phi_{1,0} && \phi_{1,0}=\phi_{0,1}.
\end{align*}
Thus, we are left with a single parameter
$\phi\coloneqq\phi_{1,0}=\phi_{0,1}=\frac{\phi_{1,1}}{2}$. Defining:
\begin{align*}
Q^{0}_\phi=I, \qquad Q^{1}_\phi=\begin{pmatrix}e^{i\phi}&0\\0&e^{-i\phi}\end{pmatrix},
\end{align*}
we can write $U_{i,j}=Q^{i}_\phi Q^{j}_\phi$. Recalling
\Eq\eqref{eq:decz}, we get to the form:
\begin{align*}
&\tT(\xi)=(C_\phi\gt I)(I\gt C_\phi)(I\gt\xi\gt I)(I\gt C_\phi)^\dag (C_\phi\gt I)^\dag,
\end{align*}
where $C_{\phi}$ is as in \Eq\eqref{eq:dephasing}. Finally, including
the invariance under local transformation we can get to the general
form:
\begin{align*}
\resizebox{\hsize}{!}{$
\tT(\xi)=(C_\phi\gt I)(I\gt C_\phi)(I\gt U\xi U^\dag \gt I)(I\gt C_\phi)^\dag (C_\phi\gt I)^\dag.$}
\end{align*}
As the proof holds for a generic odd operator, it holds in particular
for a representation of the CAR $(\xi_{2i},\xi_{2i+1})$,
and this concludes the proof.
\end{proof}

\section{Proof of Theorem~\ref{thm:main}}\label{sec:mthm}
We want now to prove the Main Theorem \ref{thm:main}. The proof
proceeds as follows: first, we find in Proposition \ref{prop:form}
the most generic FCA with $\EL$ and $\ER$ generated by odd
anticommuting operators, whose evolution is constrained by two
mutually exclusive conditions. Then in Lemma \ref{lem:nonhopiunomi}
we prove that only one of the two conditions is compatible with the
assumptions on the support algebras $\EL,\ER$. Finally, we show that
such a condition leads to the \textit{Forking automaton} defined in
\Eq\eqref{eq:forking}.
\begin{proposition}\label{prop:form}
  Let $\tT$ be an arbitrary FCA having $\EL$ and $\ER$ generated by
  two odd anticommuting operators. The most general form of the
  evolution of the FCA acting on an odd operator $\xi$ is given by:
\begin{align}\label{eq:formT}
\tT(\xi)=b_\xi X\gt I\gt I + c_\xi I\gt I\gt Y + I \gt \m D_{\xi}\gt I,
\end{align}
whith $D_\xi$ an odd operator, along with one of the following two conditions holding:
\begin{align}\label{eq:mutcond}
D_\xi=0, \qquad b_\xi c_\xi=0.
\end{align}
\end{proposition}
\begin{proof}
  If we assume both the left and right support algebra to be generated
  by odd operators $\mathsf{G}_L,\mathsf{G}_R$, the graded commutation
  relation imposes:
\begin{equation*}
\gc{\mathsf{G}_L}{\mathsf{G}_R}= \{\mathsf{G}_L,\mathsf{G}_R\}=0.
\end{equation*}
Notice that this means that $\mathsf{G}_L$ and $\mathsf{G}_R$ are
another representation of the CAR.  Since we are interested to
classify automata modulo unitary transformations, we can take without
loss of generality $\mathsf{G}_L$ and $\mathsf{G}_R$ to be $X,Y$. Let
us consider the general decomposition for the evolution of an odd
operator $\xi$:
\begin{equation*}
\begin{split}
\tT(\xi)
=&X\gt \m A_{\xi}\gt Y + X \gt \m B_{\xi} \gt I  \\
&+I \gt  \m C_{\xi} \gt Y + I \gt \m D_{\xi} \gt I .
\end{split}
\end{equation*}
Since the evolution should be parity preserving we have:
\begin{equation*}
\m A_{\xi},\m D_{\xi} \in \m A_1,\qquad
\m B_{\xi},\m C_{\xi} \in \m A_0.
\end{equation*}
In other words, $\m A_\xi,\m D_\xi$ should be odd operators while
$\m B_{\xi},\m C_{\xi}$ should be even.  For the scope of the
proposition it is sufficient to impose the following graded
commutation relation
\begin{equation}
\gc{\tT(\xi)\gt I}{I\gt\tT(\xi)}=0.
\end{equation}
We can write the graded commutator as:
\begin{equation}\label{eq:deccom}
\begin{split}
&\gc{\tT(\xi)\gt I}{I\gt\tT(\xi)}=\\
&=X\gt \m E \gt Y + X \gt \m F \gt I + \\
&+I \gt  \m G \gt Y + I \gt \m H \gt I =0.
\end{split}
\end{equation}
Since each term is tensorized by linearly independent operators, they should identically vanish. 
Computing $\m E=0$ we get:
\begin{equation*}
\m E = \comm{\m A_\xi \gt Y}{X \gt \m A_\xi}+\m A_\xi \gt \comm{Y}{\m C_\xi} + \comm{\m  B_\xi}{X} \gt A_\xi=0.
\end{equation*}
We can now exploit the fact that the even sector of two Fermionic
modes is given by the direct sum of the two even sectors of the single
Fermionic modes and the two odd sectors, namely:
\begin{equation*}
\m A_{0,1}=(\m A^0_0\gt \m A^0_1)\oplus(\m A^1_0\gt \m A^1_1).
\end{equation*}
This means that $\m E$ decomposes as:
\begin{align*}
&\m E =(\m E_{00}, \m E_{11}) ,\\
&\m E_{00}= \comm{\m A_\xi \gt Y}{X \gt \m A_\xi},\\
&\m E_{11}=\m A_\xi \gt \comm{Y}{\m C_\xi} + \comm{\m  B_\xi}{X} \gt A_\xi,
\end{align*}
in which each component should vanish. We then have:
\begin{equation}
\comm{\m A_\xi \gt Y}{X \gt \m A_\xi}=0.
\end{equation}
Writing $\m A_\xi=\vec{p}\cdot\vec{\sigma}$ with $\vec{p}\cdot\hat{z}=0$, we get:
\begin{equation*}
p_x(\vec{p}\times \hat{y}) I\gt Z+ p_y(\vec{p}\times \hat{x}) Z\gt I=0,
\end{equation*}
but this is true \textit{iff} $\vec{p}=0$ and thus $\m A_\xi=0$ for
every even operator $\xi$, which means:
\begin{equation}\label{eq:dec}
\tT(\xi)=X \gt \m B_{\xi} \gt I + I \gt  \m C_{\xi} \gt Y + I \gt \m D_{\xi} \gt I.
\end{equation}
Notice that with this choice $\m E_{11}=0$. Computing the operators
$\m F$ and $\m G$ in \Eq\eqref{eq:deccom} for this new form of the
evolution, we have:
\begin{align*}
\m F = \comm{\m B_\xi}{X}\gt\m B_{\xi},\qquad
\m G = \m C_\xi\gt \comm{Y}{\m C_\xi}.
\end{align*}
Imposing $\m F=\m G=0$ and remembering that $\m B$ and $\m C$ are even
operators, we get $\m B_\xi,\m C_\xi \propto I$. Thus we can write the general form of the evolution as:
\begin{equation}
\tT(\xi)=b_\xi X\gt I\gt I + c_\xi I\gt I\gt Y + I \gt \m D_{\xi}\gt I,
\end{equation}
with $b_\xi,c_\xi\in\mathbb{R}$. Finally, we can compute:
\begin{equation*}
\m H= b_\xi \acomm{D_\xi}{X} \gt I + c_\xi I\gt \acomm{D_\xi}{Y}=0.
\end{equation*}
Since $D_\xi$ cannot anticommute with both $X$ and $Y$, the solutions of
the previous equations yields the constraints in
\Eq\eqref{eq:mutcond}.
\end{proof}

\begin{lemma}\label{lem:nonhopiunomi}
  The FCA $\tT$ given in \Eq\eqref{eq:formT} with $b_\xi c_\xi=0$ is
  not compatible with the assumption of having $\EL\simeq \ER$ both
  generated by odd operators.
\end{lemma}
\begin{proof}
Without loss of generality, we study the case where $b_\xi=0$. Then:
\begin{align*}
\tT^{b_\xi=0}(\xi)=I\gt D_\xi \gt I + c_\xi I\gt I\gt Y.
\end{align*}
Considering the action over another odd operator $\eta$ such that
$\acomm{\xi}{\eta}=0$ and assuming $D_\eta\neq 0$ we get to one of the
following equations:
\begin{align*}
\tT(\xi)=I\gt D_\xi \gt I + c_\xi I\gt I \gt Y,\\
\tT(\eta)= I\gt D_\eta \gt I +  c_\eta I \gt I \gt Y,
\end{align*}
or
\begin{align*}
\tT(\xi)=I\gt D_\xi \gt I + c_\xi I\gt I \gt Y,\\
\tT(\eta)= I\gt D_\eta \gt I +  b_\eta X \gt I \gt I.
\end{align*}
The first case is ruled out by the choice $\EL\simeq \ER$ because
$\EL\simeq I$, while $\ER$ is generated by Y. Let us focus on the
second case.  Lemma \ref{thm:suppcomm} tells us that if
$\gc{\tT(\xi)}{\tT(\eta)}=0$ then the whole support algebras should
graded commute. This means that we can combine the elements of
$\tT(\xi)$ and $\tT(\eta)$ and impose the graded-commutation of the
operators
\begin{align*} 
(I\gt D_\xi\gt I)c_\xi(I\gt I\gt Y)=c_\xi(I\gt D_\xi\gt Y),\\
(I\gt D_\eta\gt I)b_\eta(X\gt I\gt I)=b_\eta(X\gt D_\eta\gt I)
\end{align*}
i.e.:
\begin{equation*}
\gc{I\gt D_\xi \gt Y}{X\gt D_\eta\gt I}=X\gt\comm{D_\xi}{D_\eta}\gt Y=0.
\end{equation*}
This means $D_\xi=D_\eta\eqqcolon D$. Imposing
$\gc{\tT(\xi)\gt I}{I\gt\tT(\xi)}=\gc{\tT(\eta)\gt
  I}{I\gt\tT(\eta)}=0$ and exploiting again Lemma \ref{thm:suppcomm}
for the operator of the support $\gc{I\gt D}{ I\gt Y}$ and
$\gc{X\gt I}{D\gt I}$ we have:
\begin{align*}
D\gt\acomm{D}{Y}=\acomm{D}{X}\gt D=0.
\end{align*}
This is clearly compatible only with $D=0$. Finally, we want to analyze the case:
\begin{align*}
\tT(\xi)=I\gt D_\xi \gt I + c_\xi I\gt I \gt Y\\
\tT(\eta)= c_\eta I \gt I\gt Y +  b_\eta X \gt I \gt I.
\end{align*}
Imposing
$\gc{\tT(\xi)\gt I}{I\gt\tT(\eta)}=\gc{\tT(\eta)\gt
  I}{I\gt\tT(\xi)}=0$ and exploiting again Lemma \ref{thm:suppcomm}
for the operators
\begin{align*}\resizebox{\hsize}{!}{$
c_\xi b_\xi\gc{I\gt I\gt D_\eta}{X\gt I \gt Y},\qquad \gc{ D_\eta \gt I\gt I }{X\gt I \gt Y}$}
\end{align*}
we have:
\begin{align*}
X\gt\acomm{D_\eta}{Y}=\acomm{D_\eta}{X}\gt Y=0.
\end{align*}
Again, this is only possible if $D_\eta=0$, which concludes the proof.
\end{proof}
We are ready to prove the Main Theorem \ref{thm:main}.
\begin{proof}
Thanks to the last lemma, we can impose $D_\xi=D_\eta=0$ in \Eq\eqref{eq:formT}, i.e.
\begin{align*}
\tT(\xi)=b_\xi X\gt I\gt I+ c_\xi I\gt I\gt Y,\\
\tT(\eta)=b_\eta X\gt I\gt I+ c_\eta I\gt I\gt Y.
\end{align*}
Imposing $\gc{\tT(\xi)}{\tT(\eta)}=0$, we get:
\begin{align*}
(b_\xi b_\eta+c_\xi c_\eta) I\gt I \gt I=0.
\end{align*}
This means that if we arrange $b,c$ in vectors
\begin{align*}
v_{\xi/\eta}=\begin{pmatrix}b_{\xi/\eta}\\c_{\xi/\eta}\end{pmatrix},
\end{align*}
then it holds that $v_\xi\cdot v_\eta=0$. However, we can always choose
$\xi,\eta$ in such a way that $v_\xi=\begin{pmatrix}1\\0\end{pmatrix}$
and $v_\eta=\begin{pmatrix}0\\1\end{pmatrix}$ so that:
\begin{align*}
\tT(\xi)= I\gt I\gt Y,\qquad \tT(\eta)= X\gt I\gt I,
\end{align*}
and this concludes the proof.
\end{proof}

\subsection{Proof of Corollary~\ref{cor:Fork}}\label{sec:Cor}

\begin{proof}
Consider the Forking automaton $\tT_0$ that acts over $X,Y$, i.e.:
\begin{equation*}
I\gt\tT(X,Y)\gt I= (X\gt I \gt I,I\gt I \gt Y).
\end{equation*}
This is the general form in \Eq\eqref{eq:forking} composed with a
rotation that brings $\eta,\xi$ in $X,Y$. In general, this acts over
$X,Y$ as:
\begin{align*}
X_i\rightarrow X_{i-1},\qquad Y_i\rightarrow Y_{i+1}.
\end{align*} 
Since we are interested in the classification modulo factorized
unitaries, we can consider the equivalent evolution:
\begin{align*}
X_i\rightarrow Y_{i-1},\qquad
Y_i\rightarrow X_{i+1}.
\end{align*} 
We can split this action into two subsequent steps. The first one
swaps $X_{2i}$ with $Y_{2i-1}$ and the second swaps $X_{2i+1}$ with
$Y_{2i}$, i.e.
\begin{align*}
\tM_1: \begin{cases}
I\gt X_{2i}\rightarrow Y_{2i-1}\gt I\\
I\gt Y_{2i}\rightarrow I\gt Y_{2i}\\
X_{2i-1}\gt I\rightarrow X_{2i-1}\gt I\\
Y_{2i-1}\gt I\rightarrow I\gt X_{2i}.\end{cases}\\
\tM_2:\begin{cases}
X_{2i}\gt I\rightarrow X_{2i}\gt I\\
Y_{2i}\gt I \rightarrow I\gt X_{2i+1}\\
I\gt X_{2i+1}\rightarrow Y_{2i}\gt I\\
I\gt Y_{2i+1}\rightarrow I\gt Y_{2i+1}.\end{cases}\\
\end{align*}
It is easy to see that $\tM_1=\tau_{\pm 1}\circ \tM_2$. Thus, the only
thing we need to show to prove the local implementability of this
particular Forking is that the action $\tM_1$
can be implemented through
conjugation with a unitary operator. Since the Margolus scheme
is not unique, it is sufficient to show an example of such an
operator. One implementation is given by
\begin{align*}
  &\tM_1(T)=(X\gt Z)e^{-\frac{\pi}{4}Y\gt X}Te^{\frac{\pi}{4}Y\gt X}(X\gt Z).
\end{align*}
Indeed,  let $\Sigma_i\coloneqq X_{i},Y_{i}$. Writing
$e^{-\frac{\pi}{4}Y_{2i}\gt X_{2i+1}}=\frac{1}{\sqrt{2}}\{I\gt I - 
Y_{2i}\gt X_{2i+1}\}$ we have:
\begin{align*}
  &\frac{1}{2}\{I\gt I -Y_{2i}\gt X_{2i+1}\}(\Sigma_{2i}\gt I)\{I\gt I + Y_{2i}\gt X_{2i+1}\}=\\
&=\frac{1}{2}\{I\gt I - Y_{2i}\gt X_{2i+1}\}\{\Sigma_{2i}\gt I+\Sigma_{2i}Y_{2i}\gt X_{2i+1}\}=\\
  &=\frac{1}{2}\{ \Sigma_{2i}\gt I - Y_{2i}\Sigma_{2i} Y_{2i}\gt I+ \{Y_{2i},\Sigma_{2i}\}\gt X_{2i+1}\}\\
\end{align*}
Notice that we exploited the fact that
$(Y\gt X)\Sigma\gt I (Y\gt X)=(Y\gt X)(I\gt X)\Sigma\gt I(Y\gt
I)$ that trivially follows from the fact that two elements of the CAR
representation on different sites always anticommute.  For
$\Sigma_{2i}=X_{2i+1}$ this gives $X_{2i}\gt I$ and for $\Sigma_{2i}=Y_{2i}$ we have $I\gt X_{2i}$ as desired. Repeating the computation for $I\gt \Sigma_{2i+1}$ we get
$X_{2i+i}\mapsto$ and $Y_{2i+1}\mapsto Y_{2i+1}$ while $X_{2i+1}\mapsto -Y_{2i}$. Conjugation with $X_{2i}\gt Z_{2i+1}$ finally fixes the sign so that $X_{2i+1}\mapsto Y_{2i}$.
To prove that this holds also for the general form of the
Forking automaton $\tT'_0$ we can write it as:
\begin{align*}
\tT'_0=\tT_0\circ\tU,
\end{align*}
where $\tT_0$ is the forking automaton we just considered and $\tU$ is
a local change of basis, i.e. the one that brings $\xi,\eta$ into
$X,Y$. Again, since we are interest in the classification modulo local
rotation we can always conjugate with a suitable unitary to write the
general form of $\tT'_0$ as:
\begin{align*}
\tT'_0=\tU\circ\tT_0=\tU\circ\tM_2\circ\tM_1.
\end{align*}
Including the rotation in the gates of $\tM_2$ we conclude the proof.
\end{proof}
\end{document}

%% file: tikzstyle.tex
\usetikzlibrary {arrows.meta}
\usetikzlibrary{%
    decorations.pathreplacing,%
    decorations.pathmorphing,%
    arrows,
    arrows.meta,
    positioning,
    calc,
    shapes,
    shadows,
    shapes.geometric,
   math
    }

\definecolor{myblue1}{RGB}{0,105,189}
\definecolor{myblue2}{RGB}{0,155,249}
\definecolor{mycolor}{RGB}{100,155,239}
\definecolor{cifcolor}{RGB}{140,140,20}
\definecolor{bitcol}{RGB}{0,125,209}
\definecolor{ancillacol}{RGB}{200,130,50}
\definecolor{Ucolor}{RGB}{150,20,140}
\definecolor{ancillamod}{RGB}{100,180,50}
\definecolor{bitmod}{RGB}{90,85,209}

\tikzmath{
let \bitsize = 0.5 cm;
let \bitdistance =1 cm;
let \transfborder =0.2 cm;
let \transfsize=\bitsize+\bitdistance-\transfborder;
let \transfsizeanc=2*\bitsize+2*\bitdistance-\transfborder;
let \transfdistance=\transfsizeanc+\bitdistance*4;
let \transfoffset=\transfsizeanc-4*\transfborder;
let \transfshift=\transfdistance+\transfsizeanc;
let \halfsize=\bitsize*0.5;
let \halfdistance=\bitdistance*0.5;
let \hdist=2*\bitsize+\bitdistance;
let \nhdist=-\bitsize*1.5-\bitdistance;
let \halfleft=-\bitsize-\transfoffset + \transfsize*0.5-\transfborder;
 }
\tikzset{
  local/.style 2 args={
    rectangle, 
    rounded corners=1pt, 
    minimum height={#1},
    minimum width={#2},
    align=center, 
    line width=1pt,
    draw=white,
    font=\color{black}\sffamily, 
    fill=myblue1
  },
  local/.default={0.6 cm}{0.6 cm},
  Manc1/.style 2 args={
 	rectangle,
 	rounded corners=5pt ,
 	align=center,
 	minimum height={#1},
   	 minimum width={#2},
 	line width =1pt,
 	draw=black,
    font=\color{black}\sffamily, 
    fill=myblue1,
  },
  Manc1/.default={\bitsize}{\transfsizeanc},
 Manc2/.style 2 args={
 	rectangle,
 	rounded corners=5pt ,
 	align=center,
 	minimum height={#1},
   	 minimum width={#2},
 	line width =1pt,
 	draw=black,
    font=\color{black}\sffamily, 
    fill=myblue2,
  },
  Manc2/.default={\bitsize}{\transfsizeanc},
 G/.style 2 args={
 	rectangle,
 	rounded corners=2pt ,
 	align=center,
 	minimum height={#1},
   	 minimum width={#2},
 	line width =1pt,
 	draw=black,
    font=\color{black}\sffamily, 
    fill=mycolor
  },
  G/.default={\bitsize}{\transfsize},
  anc/.style= {  
    rectangle, 
    rounded corners=5pt, 
    minimum height={\bitsize},
    minimum width={\bitsize},
    align=center, 
    line width=1pt,
    draw=black,
    font=\color{black}\sffamily, 
    fill=ancillacol
   },
   qbitmod/.style={
    rectangle, 
    rounded corners=5pt, 
    minimum height={\bitsize},
    minimum width={\bitsize},
    align=center, 
    line width=1pt,
    draw=black,
    font=\color{black}\sffamily, 
    fill=bitmod
   },
     ancmod/.style={
    rectangle, 
    rounded corners=5pt, 
    minimum height={\bitsize},
    minimum width={\bitsize},
    align=center, 
    line width=1pt,
    draw=black,
    font=\color{black}\sffamily, 
    fill=ancillamod
   },
  qbit/.style= {  
    rectangle, 
    rounded corners=5pt, 
    minimum height={\bitsize},
    minimum width={\bitsize},
    align=center, 
    line width=1pt,
    draw=black,
    font=\color{black}\sffamily, 
    fill=white
   },
   qbit2/.style={  
    rectangle, 
    rounded corners=5pt, 
    minimum height={\bitsize},
    minimum width={\bitsize},
    align=center, 
    line width=1pt,
    draw=black,
    font=\color{black}\sffamily, 
    fill=bitcol
   },
  U/.style= {  
    rectangle, 
    rounded corners=1pt, 
    minimum height={\bitsize},
    minimum width={\bitsize},
    align=center, 
    line width=1pt,
    draw=black,
    fill=Ucolor,
    font=\color{black}\sffamily, 
   },
  LRbit/.style = {  
    rectangle, 
    rounded corners=1pt, 
    minimum height={\bitsize},
    minimum width={\transfsize},
    align=center, 
    line width=1pt,
    draw=black,
    fill=gray,
    font=\color{black}\sffamily, 
   },
   M1/.style 2 args={
 	rectangle,
 	rounded corners=5pt ,
 	align=center,
 	minimum height={#1},
   	 minimum width={#2},
 	line width =1pt,
 	draw=black,
    font=\color{black}\sffamily, 
    fill=myblue1,
  },
  M1/.default={\bitsize}{\transfsize},
  M2/.style 2 args={
 	rectangle,
 	rounded corners=5pt ,
 	align=center,
 	minimum height={#1},
   	 minimum width={#2},
	draw=black,
 	line width =1pt,
    font=\color{black}\sffamily, 
    fill=myblue2,
  },
  M2/.default={\bitsize}{\transfsize},
  A/.style 2 args={
 	rectangle,
 	rounded corners=5pt ,
 	align=center,
 	minimum height={#1},
   	 minimum width={#2},
 	line width =1pt,
 	draw=black,
	opacity=0.9,
    font=\color{black}\sffamily, 
    fill=ancillacol,
  },
  A/.default={\bitsize}{\bitsize},
LR/.style ={
 	rectangle split,
	minimum height=0.6 cm,
   	minimum width =7.5cm,
	rectangle split horizontal,
	rectangle split parts=2,
 	rounded corners=2pt ,
 	align=center,
 	line width =1pt,
 	draw=black,
	opacity=0.6,
    font=\color{black}\sffamily, 
    fill=yellow,
 	text centered,
  },
  swap/.pic ={
	\draw (-\bitdistance,-1)--(-\bitdistance,-0.75)--(\bitdistance,0.75)--(\bitdistance,1);
	\draw (\bitdistance,-1)--(\bitdistance,-0.75)--(-\bitdistance,0.75)--(-\bitdistance,1);  
  },
Cif/.style 2 args={
 	rectangle,
 	rounded corners=2pt ,
 	align=center,
 	minimum height={#1},
   	 minimum width={#2},
 	line width =1pt,
 	draw=black,
    font=\color{black}\sffamily, 
    fill=cifcolor
  },
Cif/.default={\bitsize}{\transfsize},
Op/.style 2 args={
	circle,
	draw=black,
	line width=1pt,
	font=\color{black}\sffamily ,
	minimum height={#1},
   	 minimum width={#2},
   	 fill=White,
 },
 Op/.default={\bitsize*0.25}{\bitsize*0.25},
 J/.pic={
  \draw (-\bitdistance,\nhdist)--(0,\nhdist*0.5)--(\bitdistance,\nhdist),},
 }

%% file: sample.tikzstyles
\usetikzlibrary {arrows.meta}
\usetikzlibrary{%
    decorations.pathreplacing,%
    decorations.pathmorphing,%
    arrows,
    arrows.meta,
    positioning,
    calc,
    shapes,
    shadows,
    shapes.geometric
    }
    
\tikzset{
  bit/.style={
  rectangle,
  minimum height=0.4cm,
  minimum width=0.4cm,
  draw=black,
  line width=1pt
  }
}
\tikzstyle{new edge style 0}=[-, fill=white, draw={rgb,255: red,249; green,20; blue,0}]
\tikzstyle{new edge style 1}=[-, fill=black]
\tikzstyle{new edge style 2}=[-, fill={rgb,255: red,255; green,247; blue,0}, draw={rgb,255: red,162; green,164; blue,60}]
\tikzstyle{new edge style 3}=[-, fill={rgb,255: red,131; green,131; blue,131}]
\tikzstyle{new edge style 4}=[-, fill={rgb,255: red,0; green,209; blue,255}, draw={rgb,255: red,49; green,57; blue,208}]

%% file: figures/Ancilla.tikz.tex
\begin{tikzpicture}
\foreach \t in  {0,...,3}{
\draw[black]  (\bitdistance*\t+\bitsize*\t,4*\hdist) -- (\bitdistance*\t+\bitsize*\t, \bitsize*0.5+2pt) ;
\draw[black]  (4*\bitdistance+\bitdistance*\t+\bitsize*\t+\bitsize*4,4*\hdist) -- (4*\bitdistance+\bitdistance*\t+\bitsize*\t+\bitsize*4, \bitsize*0.5+2pt) ;
}
\foreach \t in  {0,...,3}{
\node[qbit]  at (\bitdistance*2*\t+\bitsize*2*\t ,0)   {};
\node[anc]  at (\bitdistance*2*\t +\bitdistance+\bitsize*2*\t+\bitsize , 0)   {};
}
\foreach \t in {0,...,1}{
\node[Manc1] at (-1.5*\transfoffset+\transfdistance*\t+2*\bitsize*\t+\bitdistance*\t,\hdist) {$M_1$};
}
\node[Manc2] at (4*\bitdistance+4*\bitsize-1.5*\transfoffset,3*\hdist) {$M_2$};

\fill[myblue2,rounded corners=5pt](1.5*\bitsize+1.5*\bitdistance-\transfoffset+\transfdistance+\bitdistance+4*\transfborder,3*\hdist-\bitsize)-- (-0.5*\bitdistance-0.5*\bitsize-0.5*\transfoffset+\transfdistance+2*\bitsize+\bitdistance,3*\hdist-\bitsize)--(-0.5*\bitdistance-0.5*\bitsize-0.5*\transfoffset+\transfdistance+2*\bitsize+\bitdistance,3*\hdist+\bitsize)--(1.5*\bitsize+1.5*\bitdistance-\transfoffset+\transfdistance+\bitdistance+4*\transfborder,3*\hdist+\bitsize);
\draw[black,line width=1pt,rounded corners=5pt](1.5*\bitsize+1.5*\bitdistance-\transfoffset+\transfdistance+\bitdistance+4*\transfborder,3*\hdist-\bitsize)-- (-0.5*\bitdistance-0.5*\bitsize-0.5*\transfoffset+\transfdistance+2*\bitsize+\bitdistance,3*\hdist-\bitsize)--(-0.5*\bitdistance-0.5*\bitsize-0.5*\transfoffset+\transfdistance+2*\bitsize+\bitdistance,3*\hdist+\bitsize)--(1.5*\bitsize+1.5*\bitdistance-\transfoffset+\transfdistance+\bitdistance+4*\transfborder,3*\hdist+\bitsize);
\fill[myblue2,rounded corners=5pt] (-0.5*\bitsize-0.5*\bitdistance,3*\hdist-\bitsize)--(1.5*\bitsize+1.5*\bitdistance,3*\hdist-\bitsize)--(1.5*\bitsize+1.5*\bitdistance,3*\hdist+\bitsize)--(-0.5*\bitsize-0.5*\bitdistance,3*\hdist+\bitsize);
\draw[black,line width=1pt,rounded corners=5pt](-0.5*\bitsize-0.5*\bitdistance,3*\hdist-\bitsize)--(1.5*\bitsize+1.5*\bitdistance,3*\hdist-\bitsize)--(1.5*\bitsize+1.5*\bitdistance,3*\hdist+\bitsize)--(-0.5*\bitsize-0.5*\bitdistance,3*\hdist+\bitsize);
\end{tikzpicture}

%% file: figures/Ancillamod.tikz.tex
\begin{tikzpicture}
\foreach \t in  {0,...,3}{
\draw[black]  (\bitdistance*\t+\bitsize*\t,4*\hdist) -- (\bitdistance*\t+\bitsize*\t, \bitsize*0.5+2pt) ;
\draw[black]  (4*\bitdistance+\bitdistance*\t+\bitsize*\t+\bitsize*4,4*\hdist) -- (4*\bitdistance+\bitdistance*\t+\bitsize*\t+\bitsize*4, \bitsize*0.5+2pt) ;
}
\foreach \t in  {0,...,3}{
\node[qbit2]  at (\bitdistance*2*\t+\bitsize*2*\t ,0)   {};
\node[anc]  at (\bitdistance*2*\t +\bitdistance+\bitsize*2*\t+\bitsize , 0)   {};
}
\end{tikzpicture}

%% file: figures/Gateancilla.tikz.tex
\begin{tikzpicture}
\foreach \t in  {0,...,1}{
\draw[black]  (\bitdistance*\t+\bitsize*\t,2*\hdist) -- (\bitdistance*\t+\bitsize*\t, \bitsize*0.5+2pt) ;
\draw[black]  (2*\bitdistance+\bitdistance*\t+\bitsize*\t+\bitsize*2,2*\hdist) -- (2*\bitdistance+\bitdistance*\t+\bitsize*\t+\bitsize*2, \bitsize*0.5+2pt) ;
}
\foreach \t in  {0,...,1}{
\node[qbit]  at (\bitdistance*2*\t+\bitsize*2*\t ,0)   {};
\node[anc]  at (\bitdistance*2*\t +\bitdistance+\bitsize*2*\t+\bitsize , 0)   {};
}
\foreach \t in {0}{
\node[Manc1] at (-1.5*\transfoffset+\transfdistance*\t+2*\bitsize*\t+\bitdistance*\t,\hdist) {$M_1$};
}
\end{tikzpicture}

%% file: figures/Gateactionancilla.tikz.tex
\begin{tikzpicture}
\foreach \t in  {0,...,1}{
\draw[black]  (\bitdistance*\t+\bitsize*\t,2*\hdist) -- (\bitdistance*\t+\bitsize*\t, \bitsize*0.5+2pt) ;
\draw[black]  (2*\bitdistance+\bitdistance*\t+\bitsize*\t+\bitsize*2,2*\hdist) -- (2*\bitdistance+\bitdistance*\t+\bitsize*\t+\bitsize*2, \bitsize*0.5+2pt) ;
}
\foreach \t in  {0,...,1}{
\node[qbitmod]  at (\bitdistance*2*\t+\bitsize*2*\t ,0)   {};
\node[ancmod]  at (\bitdistance*2*\t +\bitdistance+\bitsize*2*\t+\bitsize , 0)   {};
}
\end{tikzpicture}

%% file: figures/Ancillaremcircuit0.tikz.tex
\begin{tikzpicture}
\foreach \t in  {0,...,11}{
\draw (\bitdistance*\t+\bitsize*\t ,\bitsize*0.5) -- (\bitdistance*\t+\bitsize*\t ,\hdist*6);
\node[qbit]  at (\bitdistance*\t+\bitsize*\t ,0)   {};
\node at (\bitdistance*\t+\bitsize*\t ,\nhdist) {\t};
}
\node[M1] at (\bitdistance*4.5+\bitsize*4.5,\hdist) {$M_1$};
\node[M1] at (\bitdistance*6.5+\bitsize*6.5,\hdist) {$M_1$};
\node[M1] at (\bitdistance*2.5+\bitsize*2.5,\hdist) {$M_1$};
\node[A] at (0,\hdist){$L^4$};
\node[A] at (\bitdistance+\bitsize,\hdist){$R^7$};
\node[A] at(\bitdistance*8+\bitsize*8,\hdist){$ R^5$};
\node[A] at(\bitdistance*11+\bitsize*11,\hdist){$ L^2$};
\node[A] at(\bitdistance*10+\bitsize*10,\hdist){$ R^3$};
\node[A] at(\bitdistance*9+\bitsize*9,\hdist){$ L^6$};
\end{tikzpicture}

%% file: figures/Ancillaremcircuit1.tikz.tex
\begin{tikzpicture}
\foreach \t in  {0,...,11}{
\draw (\bitdistance*\t+\bitsize*\t ,\bitsize*0.5) -- (\bitdistance*\t+\bitsize*\t ,\hdist*6);
\node[qbit]  at (\bitdistance*\t+\bitsize*\t ,0)   {};
\node at (\bitdistance*\t+\bitsize*\t ,\nhdist) {\t};
}
\node[M1] at (\bitdistance*4.5+\bitsize*4.5,\hdist) {};
\node[M1] at (\bitdistance*6.5+\bitsize*6.5,\hdist) {};
\node[M1] at (\bitdistance*2.5+\bitsize*2.5,\hdist) {};
\node[M2] at (\bitdistance*5.5+\bitsize*5.5,3*\hdist) {$M_2$};
\node[M2] at (\bitdistance*3.5+\bitsize*3.5,3*\hdist) {$M_2$};
\node[A] at (\bitdistance+\bitsize,\hdist){$R^7$};
\node[A] at(\bitdistance*11+\bitsize*11,\hdist){$ L^2$};
\end{tikzpicture}

%% file: figures/Ancillaremcircuit2.tikz.tex
\begin{tikzpicture}
\foreach \t in  {0,...,11}{
\draw (\bitdistance*\t+\bitsize*\t ,\bitsize*0.5) -- (\bitdistance*\t+\bitsize*\t ,\hdist*6);
\node[qbit]  at (\bitdistance*\t+\bitsize*\t ,0)   {};
\node at (\bitdistance*\t+\bitsize*\t ,\nhdist) {\t};
}
\node[M1] at (\bitdistance*4.5+\bitsize*4.5,\hdist) {};
\node[M1] at (\bitdistance*6.5+\bitsize*6.5,\hdist) {};
\node[M1] at (\bitdistance*2.5+\bitsize*2.5,\hdist) {};
\node[M2] at (\bitdistance*5.5+\bitsize*5.5,3*\hdist) {};
\node[M1] at (\bitdistance*8.5+\bitsize*8.5,\hdist) {$M_1$};
\node[M2] at (\bitdistance*3.5+\bitsize*3.5,3*\hdist) {};
\node[M2] at (\bitdistance*7.5+\bitsize*7.5,3*\hdist) {$M_2$};
\node[A] at (\bitdistance*3+\bitsize*3,5*\hdist){$\tiny R^9$};
\node[A] at(\bitdistance*11+\bitsize*11,\hdist){$ L^2$};
\end{tikzpicture}

%% file: figures/Ancillaremcircuit3.tikz.tex
\begin{tikzpicture}
\foreach \t in  {0,...,11}{
\draw (\bitdistance*\t+\bitsize*\t ,\bitsize*0.5) -- (\bitdistance*\t+\bitsize*\t ,\hdist*6);
\node[qbit]  at (\bitdistance*\t+\bitsize*\t ,0)   {};
\node at (\bitdistance*\t+\bitsize*\t ,\nhdist) {\t};
}
\node[M1] at (\bitdistance*4.5+\bitsize*4.5,\hdist) {};
\node[M1] at (\bitdistance*6.5+\bitsize*6.5,\hdist) {};
\node[M1] at (\bitdistance*2.5+\bitsize*2.5,\hdist) {};
\node[M2] at (\bitdistance*5.5+\bitsize*5.5,3*\hdist) {};
\node[M1] at (\bitdistance*8.5+\bitsize*8.5,\hdist) {};
\node[M2] at (\bitdistance*3.5+\bitsize*3.5,3*\hdist) {};
\node[M2] at (\bitdistance*7.5+\bitsize*7.5,3*\hdist) {};
\node[M1] at (\bitdistance*0.5+\bitsize*0.5,\hdist) {$M_1$};
\node[M2] at (\bitdistance*1.5+\bitsize*1.5,3*\hdist) {$M_2$};
\node[A] at (\bitdistance*3+\bitsize*3,5*\hdist){$\tiny R^9$};
\node[A] at (\bitdistance*4+\bitsize*4,5*\hdist){$\tiny L^0$};
\end{tikzpicture}

%% file: figures/Ancillaremcircuit4.tikz.tex
\begin{tikzpicture}
\foreach \t in  {0,...,11}{
\draw (\bitdistance*\t+\bitsize*\t ,\bitsize*0.5) -- (\bitdistance*\t+\bitsize*\t ,\hdist*6);
\node[qbit]  at (\bitdistance*\t+\bitsize*\t ,0)   {};
\node at (\bitdistance*\t+\bitsize*\t ,\nhdist) {\t};
}
\node[M1] at (\bitdistance*4.5+\bitsize*4.5,\hdist) {};
\node[M1] at (\bitdistance*6.5+\bitsize*6.5,\hdist) {};
\node[M2] at (\bitdistance*5.5+\bitsize*5.5,3*\hdist) {};
\node[M1] at (\bitdistance*8.5+\bitsize*8.5,\hdist) {};
\node[M2] at (\bitdistance*7.5+\bitsize*7.5,3*\hdist) { };
\node[M1] at (\bitdistance*2.5+\bitsize*2.5,\hdist) {};
\node[M2] at (\bitdistance*3.5+\bitsize*3.5,3*\hdist) { };
\node[M1] at (\bitdistance*0.5+\bitsize*0.5,\hdist) {};
\node[M2] at (\bitdistance*1.5+\bitsize*1.5,3*\hdist) {};
\node[M1] at (\bitdistance*10.5+\bitsize*10.5,\hdist) {$M_1$};
\node[M2] at (\bitdistance*9.5+\bitsize*9.5,3*\hdist) {$M_2$ };

\fill[myblue2, rounded corners=5pt]  (-\bitdistance*0.5-\bitsize*0.5,\bitsize+3*\hdist)--(\bitdistance*0.5+\bitsize*0.5,\bitsize+3*\hdist)--(\bitdistance*0.5+\bitsize*0.5,-\bitsize+3*\hdist)--(-\bitdistance*0.5-\bitsize*0.5,-\bitsize+3*\hdist);
\fill[myblue2, rounded corners=5pt](\bitdistance*11.5+\bitsize*11.5,\bitsize+3*\hdist)--(\transfborder+\bitdistance*10.5+\bitsize*10.5,\bitsize+3*\hdist)--(\transfborder+\bitdistance*10.5+\bitsize*10.5,-\bitsize+3*\hdist)--(\bitdistance*11.5+\bitsize*11.5,-\bitsize+3*\hdist);
\draw[black,rounded corners=5pt,line width=1pt]  (-\bitdistance*0.5-\bitsize*0.5,\bitsize+3*\hdist)--(\bitdistance*0.5+\bitsize*0.5,\bitsize+3*\hdist)--(\bitdistance*0.5+\bitsize*0.5,-\bitsize+3*\hdist)--(-\bitdistance*0.5-\bitsize*0.5,-\bitsize+3*\hdist);
\draw[black,rounded corners=5pt,line width=1pt] (\bitdistance*11.5+\bitsize*11.5,\bitsize+3*\hdist)--(\transfborder+\bitdistance*10.5+\bitsize*10.5,\bitsize+3*\hdist)--(\transfborder+\bitdistance*10.5+\bitsize*10.5,-\bitsize+3*\hdist)--(\bitdistance*11.5+\bitsize*11.5,-\bitsize+3*\hdist);
\node at (\bitdistance*11.5+\bitsize*11.5,3*\hdist){$M_2$};
\node at (-\bitdistance*0.5-\bitsize*0.5,3*\hdist){$M_2$};
\end{tikzpicture}